\documentclass{article}

 \usepackage[final,nonatbib]{nips_2016}

\usepackage[utf8]{inputenc} 
\usepackage[T1]{fontenc}    
\usepackage{url}            
\usepackage{booktabs}       
\usepackage{amsfonts}       
\usepackage{nicefrac}       
\usepackage{microtype}      

\usepackage{wrapfig,lipsum}

\title{Phase Retrieval via\\Incremental Truncated Wirtinger Flow}

\author{
 Ritesh Kolte\\
  Department of Electrical Engineering\\
  Stanford University\\
  Stanford, CA 94305 USA\\
  \texttt{ritesh.kolte@gmail.com} \\
   \And
  Ayfer \"{O}zg\"{u}r \\
  Department of Electrical Engineering \\
   Stanford University \\
   Stanford, CA 94305 USA\\
   \texttt{aozgur@stanford.edu} \\
}

\usepackage{times}
\usepackage{graphicx} 
\usepackage{subfigure} 

\usepackage{multirow,times}
\usepackage{bm,bbm,amsmath,amssymb,amsthm,amssymb}
\usepackage{IEEEtrantools}
\DeclareMathOperator*{\minimize}{minimize}
\newcommand{\lb}{\text{lb}}
\newcommand{\ub}{\text{ub}}
\newcommand{\SNR}{\text{SNR}}
\newcommand{\dist}{\text{dist}}
\newcommand{\init}{\text{init}}

\usepackage{algorithm2e}
\usepackage{algorithmic}

\newtheorem{lemma}{Lemma}
\newtheorem{theorem}{Theorem}

\newtheorem{proposition}{Proposition}

\begin{document}

\maketitle

\begin{abstract}
In the phase retrieval problem, an unknown vector is to be recovered given  quadratic measurements. This problem has received considerable attention in recent times. In this paper, we present an algorithm to solve a nonconvex formulation of the phase retrieval problem, that we call \emph{Incremental Truncated Wirtinger Flow}. Given random Gaussian sensing vectors, we prove that it converges linearly to the solution, with an optimal sample complexity. We also provide stability guarantees of the algorithm under noisy measurements. Performance and comparisons with existing algorithms are illustrated via numerical experiments on simulated and real data, with both random and structured sensing vectors. 
\end{abstract}

\section{Introduction}

Let $\bm{x}\in\mathbb{C}^n$ be an arbitrary unknown vector. We consider the problem of recovering $\bm{x}$, given data which is $m$ quadratic measurements \begin{equation}\label{eq:model_noiseless}y_i = |\bm{a}_i^*\bm{x}|^2,\quad 1\leq i\leq m,\end{equation}  where the sensing vectors $\bm{a}_i\in\mathbb{C}^n$ are known. The problem of solving this quadratic set of equations is quite general: special cases include the phase retrieval problem which arises in many imaging applications due to physical limitations of sensors, as well as the NP-complete stone problem in combinatorial optimization. 

Despite the hardness, impressive theoretical and empirical performance guarantees have been obtained recently by making some assumptions about the model, such as assuming the sensing vectors to be i.i.d. samples from, say $\mathcal{N}(0,I)$ or $\mathcal{CN}(0,I)$. The first work in this direction was \cite{CanStrVor13}, which employed the squared loss function and attempted to solve
\begin{equation}\label{eq:optim_prob}\minimize\limits_{\bm{z}\in\mathbb{R}^n} \frac{1}{m}\sum_{i=1}^m\ell(y_i,|\bm{a}_i^*\bm{z}|^2),\end{equation}
where $\ell(y_i,|\bm{a}_i^*\bm{z}|^2) = (y_i - |\bm{a}_i^*\bm{z}|^2)^2$. The non-convexity of this problem was addressed by performing a lifting step (expressing the problem in terms of $\bm{z}\bm{z}^T$) followed by a convex relaxation. This algorithm recovers $\bm{x}$ with high probability if $m=O(n)$, \cite{Can14}. However, the memory requirements and computational complexity make it prohibitive for problems of large dimension. 

Follow-up works such as \cite{CanXiaSol15}, \cite{NetJaiSan15}, \cite{CheCan15} made progress on the computational complexity front. All of these works developed algorithms to solve the problem \eqref{eq:optim_prob} directly without performing the lifting step, using a first-order optimization method for iterative refinement, after an appropriate initialization. The work \cite{CheCan15} also provided stability guarantees if the measurements are corrupted with noise. 

However, each iteration in these algorithms requires one pass through the entire data. This can be highly undesirable when the dimensions of the problem are large, since a single update can require a large amount of time, in part due to the communication delays introduced when the entire data does not fit in the available memory. Large dimensions naturally arise in the phase retrieval problem since the object of interest usually represents an image, so $n$ is the product of the image dimensions.

In this paper, we build on the idea of Truncated Wirtinger Flow (TWF) from \cite{CheCan15} and modify it to obtain the Incremental Truncated Wirtinger Flow (ITWF). By \emph{incremental}, we mean that each iteration of the algorithm only accesses one randomly chosen data point, i.e. one sensing vector and the corresponding measurement. Thus, each iteration of ITWF is cheaper than that of TWF by a factor $m$, similar to what happens, e.g., by going from full gradient descent to stochastic gradient descent (SGD) in the case of standard empirical risk minimization problems. Unfortunately, this benefit is not obtained readily for the problem at hand, since the truncation performed at each iteration in \cite{CheCan15} makes use of a threshold that is a function of all the sensing vectors and measurements. Excluding the truncation leads to a severe hit in the performance as observed in \cite{CheCan15}. Furthermore, similar to SGD, sampling one data point instead of all data points introduces variance in the descent direction. The main contribution of this paper is the design of the incremental method ITWF that matches the excellent performance of TWF in terms of the statistical complexity, computational complexity and robustness to noisy measurements despite being an incremental method. In fact, our numerical experiments demonstrate that ITWF far surpasses TWF on the computational complexity front. 

\begin{table*}[!ht]
\centering
\begin{tabular}{c | c | c  }
\hline
Algorithm & Sample complexity $m$ & Computational complexity (stage II)\\
\hline 
AltMinPhase \cite{NetJaiSan15} & $O(n\log^3n)$ & $O(n^2\log^3(n/\epsilon))$ \\ 
\hline
Wirtinger~Flow \cite{CanXiaSol15} & $O(n\log n)$ & $O\left(n^3\log(n/\epsilon)\right)$\\
\hline
Truncated Wirtinger Flow \cite{CheCan15} & $O(n)$  & $O\left(n^2\log(1/\epsilon)\right)$\\
\hline
This paper & $O(n)$ & $O\left(n^2\log(1/\epsilon)\right)$\\
\hline
\end{tabular}
\caption{Performance of Algorithms}
\end{table*}

Table 1 provides a comparison of algorithms. 
Though our original intention was to simply develop an incremental version of TWF to allow efficient handling of large data, we find that ITWF provides a remarkable speedup as compared to TWF, even at $n\approx 1000$. To provide the reader an idea about the speed-up, we mention here some observations from the numerical experiments in Section~\ref{sec:examples}. In Example~2, after an initialization stage requiring 10 passes through the data, the second stage of TWF requires at least 120 further passes to get a high accuracy solution. In contrast, the second stage of ITWF recovers a high accuracy solution in less than 15 passes. Example~3 presents an even more compelling case. Here, after 50 passes through the data for initialization, the solution returned by TWF after 50 further passes can instead be obtained by making 3 passes using ITWF.

\subsection{Related Work}
There has also been a lot of recent work on different formulations of the phase retrieval problem. In the sparse phase retrieval problem, it is additionally assumed that the vector $x$ has only a few non-zero entries, and algorithms are sought for which the sample complexity and computational complexity have optimal dependence on the number of non-zero entries. The interested reader is referred to works such as \cite{NetJaiSan15}, \cite{Jag13}, \cite{PedLeeRam14}. Phase retrieval under the assumption of structured sensing vectors has also been a topic of interest, e.g. coded diffraction patterns \cite{Can15}, STFT \cite{Jag15}. We also present numerical experiments in this paper demonstrating the performance of ITWF when the sensing vectors are structured. Recent developments on different tractable formulations of the phase retrieval problem have been compiled in the survey article \cite{JagEldHas15}.   
Finally, while we focus on devising an incremental update for the iterative refinement stage of TWF, the initialization stage can also be made incremental in a straightforward manner by using incremental algorithms for PCA, such as the algorithm from \cite{Sha15}.
 
\section{Main Idea}

In the following section, we first describe the algorithm TWF from \cite{CheCan15}, and then describe the proposed algorithm ITWF. 

Since we can only hope to recover the solution upto a global phase, we define $\dist(\bm{z},\bm{x})$ to be
$\min_{\varphi\in[0,2\pi)}\|e^{-j\varphi}\bm{z} -\bm{x}\|,$ and we implicitly assume in the remainder of the paper that $\bm{z}$ is $e^{-j\varphi(\bm{z})}\bm{z}$, where $\varphi(\bm{z})$ is the argument of the above minimization problem. Then, we use $\bm{h}$ to denote $\bm{z} - \bm{x}.$ Thus, $\|\bm{h}\| = \dist(\bm{z},\bm{x}).$

\subsection{Truncated Wirtinger Flow}
The loss function used in \cite{CheCan15} is $\ell(y,\hat{y}) = \hat{y} - y\log \hat{y},$ which would correspond to maximizing the log likelihood if the measurements were assumed to arise from a Poisson noise model. 
The TWF algorithm has two stages described below, and requires choosing the constants $\alpha_y$, $\mu$, $\alpha_z^{\lb}$, $\alpha_z^{\ub}$, $\alpha_h$. 

(1) Truncated Spectral Initialization\\
Set $\bm{z}^{(0)}$ to be the principal eigenvector (appropriately scaled) of 
$$\frac{1}{m}\sum_{i=1}^m y_i \bm{a}_i\bm{a}_i^*\mathbbm{1}_{\left\{y_i\leq \alpha^2_y\left(\frac{1}{m}\sum_{i=1}^m y_i\right)\right\}},$$
with $\alpha_y$ set to, say 3. The point $\bm{z}^{(0)}$ satisfies \begin{equation}\label{eq:spectral_guarantee}\dist(\bm{z}^{(0)},\bm{x}) \leq \delta \|\bm{x}\|,\end{equation} for any $\delta >0$, as long as $m/n$ exceeds a sufficiently large constant. For intuition and proof of this fact, we refer the reader to \cite{CheCan15}.

(2) Truncated Wirtinger Flow\\
For $t=0,1,\dots, T-1$, perform the update
\begin{equation}\label{eq:trunc_grad}\bm{z}^{(t+1)} = \bm{z}^{(t)} - \frac{\mu}{m}\sum_{i=1}^m\nabla \ell(y_i,|\bm{a}_i^*\bm{z}^{(t)}|^2) \cdot \mathbbm{1}_{\mathcal{E}^i_{1,t} \cap \mathcal{E}^i_{2,t}}, \end{equation}
where $\nabla \ell(y_i,|\bm{a}_i^*\bm{z}^{(t)}|^2)$ denotes the Wirtinger derivative of with respect to $\bm{z}$:
$$\nabla \ell(y_i,|\bm{a}_i^*\bm{z}^{(t)}|^2) = \frac{2|\bm{a}_i^*\bm{z}^{(t)}|^2 - 2y_i}{\bm{z}^{(t)*}\bm{a}_i}\bm{a}_i,$$ and the events $\mathcal{E}^i_{1,t}$ and $\mathcal{E}^i_{2,t}$ are defined as follows,
$$\mathcal{E}^i_{1,t} := \left\{ \alpha_z^{\lb} \leq \frac{|\bm{a}_i^*\bm{z}^{(t)}|}{\|\bm{z}^{(t)}\|} \leq \alpha_z^{\ub} \right\} ,$$
\begin{equation}\label{eq:E2}\mathcal{E}^i_{2,t} := \left\{ \left|y_i - |\bm{a}_i^*\bm{z}^{(t)}|^2\right| \leq \frac{\alpha_h}{m}\sum_{i=1}^m \left|y_i - |\bm{a}_i^*\bm{z}^{(t)}|^2\right|\right\}.\end{equation}
The idea is to perform an update similar to gradient descent. 
The truncation results in dropping those indices for which the numerator and denominator magnitudes are very different from their expected values respectively, 
since such terms can exert a large atypical influence causing the full gradient to point in an undesirable direction.

\subsection{Incremental Truncated Wirtinger Flow}

The straightforward way of turning the above algorithm into an incremental one would be to replace \eqref{eq:trunc_grad} by:
\begin{equation}\label{eq:inc_trunc_grad_first}\bm{z}^{(t+1)} = \bm{z}^{(t)} - \mu  \nabla\ell(y_{i_t},|\bm{a}_{i_t}^*\bm{z}^{(t)}|^2)\cdot \mathbbm{1}_{\mathcal{E}^{i_t}_{1,t} \cap \mathcal{E}^{i_t}_{2,t}}, \end{equation}
where $i_t$ would be chosen uniformly at random from $\{1,2,\dots, m\}$. However, as can be seen from \eqref{eq:E2}, checking if $\mathcal{E}^i_{2,t}$ has occurred for any $i$ requires a full pass through the entire data. Thus,  \eqref{eq:inc_trunc_grad_first} is as costly as \eqref{eq:trunc_grad}. To address this difficulty, we replace $\mathcal{E}^i_{2,t}$ by the event $\mathcal{E}^i_{3}$, which is defined as $$\mathcal{E}^i_{3} :=   \left\{ \frac{y_i}{\frac{1}{m}\sum_{i=1}^m y_i } \leq (\alpha_x)^2 \right\},$$ where $\alpha_x$ is a constant to be chosen. This is amenable for an incremental update, which we call Incremental Truncated Wirtinger Flow (ITWF):\begin{equation}\label{eq:inc_trunc_grad}\bm{z}^{(t+1)} = \bm{z}^{(t)} - \mu \nabla\ell(y_{i_t},|\bm{a}_{i_t}^*\bm{z}^{(t)}|^2)\cdot \mathbbm{1}_{\mathcal{E}^{i_t}_{1,t} \cap \mathcal{E}^{i_t}_{3}}. \end{equation}
Note in fact that not only does $\mathcal{E}^i_{3}$ not depend on the entire data, but it is non-adaptive, i.e. it does not depend on $\bm{z}^{(t)}$. So, $\mathcal{E}^i_{3}$ simply discards measurements in which $|\bm{a}_i^T\bm{x}|$ deviates a lot from $\|\bm{x}\|$, similar to the rule used during the initialization stage.

The main implication of \eqref{eq:E2} used in \cite{CheCan15} is that the magnitude of the term $\bm{a}_i^T\bm{h}$ can be controlled, which turns out to be crucial in obtaining the optimal sample complexity $m=O(n)$. Since ITWF cannot employ this truncation rule, we are not able to control the magnitude of $\bm{a}_i^T\bm{h}$. However, replacing \eqref{eq:E2} by $\mathcal{E}^i_{3}$ does not result in any deterioration in the sample complexity and computational complexity. 

Furthermore, we would like to point out an empirical observation that we did not observe a significant difference in the numerical experiments even if we only employed truncation based on $\mathcal{E}_{1,t}^i$, i.e. by excluding $\mathcal{E}^i_{3}$! In fact, we also observed that TWF also performs similarly whether $\mathcal{E}_{2,t}^i$ is included or not. This suggests that while the theoretical analysis of TWF and ITWF require these additional truncation events, it might be possible to prove the convergence results via a different line of analysis that does not require introducing these events. However, we observed that truncation based on $\mathcal{E}_{1,t}^i$ is indeed crucial in the numerical experiments.

\begin{algorithm}[t]
\hrule\vspace{1mm}
 \KwData{Measurements and sampling vectors $\{y_i, \bm{a}_i\}_{1\leq i\leq m}$}
 \KwResult{$\bm{z}^{(T)}$}
 \textbf{Stage I} (initialization): 
 
 Set $\bm{z}^{(0)}$ to be $\sqrt{\frac{1}{m}\sum_{i=1}^my_i}\cdot\bm{z}_{\text{init}}$, where $\bm{z}_{\text{init}}$ is the principal eigenvector of $\frac{1}{m}\sum_{i=1}^m y_i \bm{a}_i\bm{a}_i^*\mathbbm{1}_{\left\{y_i\leq \alpha^2_y\left(\frac{1}{m}\sum_{i=1}^m y_i\right)\right\}}.$
 
\textbf{Stage II}:  

\For{$t = 1,2,\dots, T-1$}{
  Sample $i_t$ uniformly at random from $\{1,2,\dots, m\}$
  $$\bm{z}^{(t+1)} = \bm{z}^{(t)} - \mu \nabla\ell(y_{i_t},|\bm{a}_{i_t}^*\bm{z}^{(t)}|^2)\cdot \mathbbm{1}_{\mathcal{E}^{i_t}_{1,t} \cap \mathcal{E}^{i_t}_{3}}$$
   }\hrule\vspace{1mm}
   \label{algo:ITWF}
 \caption{Incremental Truncated Wirtinger Flow}
\end{algorithm}

\section{Main Results and Discussion}

We focus on the real valued case for simplicity of exposition, and 
for concreteness, we fix the constants $\alpha_z^{\lb}$, $\alpha_z^{\ub}$ and $\alpha_x$ to be $0.3$, $5$ and $5$ respectively. The following two theorems are counterparts of the two main results in \cite{CheCan15}. The first result, Theorem~\ref{thm:1}, focuses on the noiseless case \eqref{eq:model_noiseless} and provides a linear convergence guarantee. 

\begin{theorem}\label{thm:1}
Under noiseless measurements \eqref{eq:model_noiseless} with $\{\bm{a}_i\}_{i=1}^m$ $\sim\mathcal{N}(0,I)$ independent, there exist universal constants $C,c_0,c_1,c_2 > 0$ and $0<\rho, \nu < 1$, such that with probability at least $1 - Cm\exp(-c_1n)$ and $\mu = c_2/n$, the iterates in Algorithm~\ref{algo:ITWF} satisfy
\begin{IEEEeqnarray}{l}
\mathbb{E}_{\mathcal{I}^t}\left[\emph{dist}^2(\bm{z}^{(t)},\bm{x})\right] \leq \nu \left(1 - \frac{\rho}{n}\right)^t\|\bm{x}\|^2,\IEEEeqnarraynumspace\label{eq:thm1}
\end{IEEEeqnarray}
if $m \geq c_0n$.
\end{theorem}

Thus, the (mean squared error) MSE is reduced by a factor $\left(1 - \rho/n\right)^m$ after one pass through the data.

In the above statement, $\mathbb{E}_{\mathcal{I}^t}[\cdot]$ denotes the expectation with respect to algorithm randomness $\mathcal{I}^t = \{i_1,i_2,\dots,i_{t-1}\}.$ More formally, it denotes the expectation conditioned on the data randomness $\{\bm{a}_i\}_{i=1}^m$. To keep expressions simple, we will follow the convention that $\mathbb{E}_{X}[\cdot]$ denotes expectation conditioned on all random variables except $X$.

Note that the expectation is only with respect to the algorithm randomness, not with respect to the data randomness $\{\bm{a}_i\}_{i=1}^m$. This is crucial since the data could be provided as it is, thus necessitating the need for a convergence guarantee that holds with high probability with respect to the data randomness. Since the convergence is linear, a simple application of Markov's inequality already provides a strong convergence guarantee in which the high probability also refers to the algorithm randomness.

Also note that linear convergence is achieved for our setup even with an incremental method since the effect of the variance of the stochastic gradient can be controlled without having to choose a step-size that decreases with iteration, while in general, incremental methods suffer from slow convergence unless variance-reducing modifications \cite{Joh13} are used. The reason for this is explained in the next subsection.

Instead of \eqref{eq:model_noiseless}, if the measurements are noisy such that
\begin{equation}\label{eq:model_noisy}y_i = |\bm{a}_i^*\bm{x}|^2 + \eta_i,\quad 1\leq i\leq m,\end{equation} where $\eta_i$ denotes the noise term (need not be stochastic), then we have the following theorem which provides a stability guarantee. 

\begin{theorem}\label{thm:2}
Under noisy measurements \eqref{eq:model_noisy} with $\{\bm{a}_i\}_{i=1}^m$ $\sim\mathcal{N}(0,I)$ independent and the noise satisfying $\|\bm{\eta}\|_{\infty} \leq \epsilon_{\eta}\|\bm{x}\|^2$ for some small constant $\epsilon_{\eta}>0$, there exist universal constants $C,c_0,c_1,c_2 > 0$ and $0<\rho, \nu < 1$, such that with probability at least $1 - Cm\exp(-c_1n)$ and $\mu = c_2/n$, the iterates in Algorithm~\ref{algo:ITWF} satisfy
\begin{IEEEeqnarray}{l}
\mathbb{E}_{\mathcal{I}^t}\left[\emph{dist}^2(\bm{z}^{(t)},\bm{x})\right] \lesssim \frac{\|\bm{\eta}\|^2}{m\|\bm{x}\|^2} + \left(1 - \frac{\rho}{n}\right)^t\|\bm{x}\|^2,\IEEEeqnarraynumspace\label{eq:thm2}
\end{IEEEeqnarray}
if $m \geq c_0n$.
\end{theorem}

As described in \cite{CheCan15}, this result can be applied to the case when the  measurements are obtained independently according to a Poisson noise model:
\begin{equation}\label{eq:poisson}y_i \sim \text{Poisson}(|\bm{a}_i^*\bm{x}|^2)\quad 1\leq i\leq m,\end{equation}
and the solution satisfies $\|\bm{x}\|^2\geq \log^3m$ (required to ensure that the condition $\|\bm{\eta}\|_{\infty} \leq \epsilon_{\eta}\|\bm{x}\|^2$ holds with high probability).

\subsection{Intuition}\label{subsec:proof_sketch}

Consider the noiseless case. The reason why we can expect linear convergence from ITWF can be understood by the following thought experiment. Say, after the $t^{\text{th}}$ iteration of ITWF, we had the luxury of obtaining a new measurement by sampling $\bm{a}_t$ independently from the population distribution $\mathcal{N}(0,I)$, instead of being restricted to sample from the empirical distribution  $\{\bm{a}_i\}_{i=1}^m$.  

Let $\bm{h}$ denote $\bm{z}^{(t)} - \bm{x}$, and $\mathcal{E}_t = \mathcal{E}_{1,t} \cap \mathcal{E}_{3} = \left\{ \alpha_z^{\lb} \leq \frac{|\bm{a}_t^*\bm{z}^{(t)}|}{\|\bm{z}^{(t)}\|} \leq \alpha_z^{\ub} \right\} \cap  \left\{ \frac{|\bm{a}_t^T\bm{x}|}{\|\bm{x}\| } \leq \alpha_x \right\}$ denote the overall truncation event. Then the expected distance to the optimal solution after performing the ITWF update \eqref{eq:inc_trunc_grad} using the new independent measurement, is
\begin{IEEEeqnarray}{rCl}
\mathbb{E}_{t}\left[\dist^2(\bm{z}^{(t+1)},\bm{x})\right] & = & \|\bm{h}\|^2 - 4\mu\mathbb{E}_t\left[ |\bm{a}_t^T\bm{h}|^2\left(1 + \frac{\bm{a}_t^T\bm{x}}{\bm{a}_t^T\bm{z}^{(t)}}\right)\mathbbm{1}_{\mathcal{E}_{t}}\right]\nonumber\\
&& \quad +\> 4\mu^2\mathbb{E}_t\left[ \|\bm{a}_t\|^2|\bm{a}_t^T\bm{h}|^2\left(1 + \frac{\bm{a}_t^T\bm{x}}{\bm{a}_t^T\bm{z}^{(t)}}\right)^2 \mathbbm{1}_{\mathcal{E}_{t}}\right]\label{eq:thought_ITWF_noiseless_rec}
\end{IEEEeqnarray}

Since $\bm{a}_t\sim\mathcal{N}(0,I)$ independent of $\bm{h},\bm{z},\bm{x}$, it is not difficult to show that if $\|\bm{h}\|$ is sufficiently small compared to $\|\bm{x}\|$,
\begin{IEEEeqnarray*}{l}
\mathbb{E}_t\left[ |\bm{a}_t^T\bm{h}|^2\left(1 + \frac{\bm{a}_t^T\bm{x}}{\bm{a}_t^T\bm{z}}\right)\mathbbm{1}_{\mathcal{E}_{t}}\right] \geq (1 - c_1)\|\bm{h}\|^2 ,\>\>
\mathbb{E}_t\left[ \|\bm{a}_t\|^2|\bm{a}_t^T\bm{h}|^2\left(1 + \frac{\bm{a}_t^T\bm{x}}{\bm{a}_t^T\bm{z}}\right)^2 \mathbbm{1}_{\mathcal{E}_{t}} \right]\lesssim n \|\bm{h}\|^2,
\end{IEEEeqnarray*}
for an appropriate constant $c_1$. As can be observed from the latter inequality, the variance of the stochastic gradient is proportional to $\|\bm{h}\|^2$. As a result, the variance is not only bounded, but reduces as we get closer to the solution, which allows us to choose a non-diminishing step-size $\mu = \Theta\left(\frac{1}{n}\right)$, resulting in linear convergence:
$$\mathbb{E}_{t}\left[\dist^2(\bm{z}^{(t+1)},\bm{x})\right] \leq \left(1 - \frac{\rho}{n}\right)\dist^2(\bm{z}^{(t)},\bm{x}).$$

In the above thought experiment, we had the luxury of an infinite amount of measurements at our disposal. Back in reality, we have the restriction of using a finite set of $m$ measurements. Since we need to keep sampling from this set to simulate the stochastic gradient descent from our thought experiment, the proof of Theorem~\ref{thm:1} mainly involves showing that we can draw similar conclusions despite the finiteness of the data set and the resulting dependence between the sensing vectors and the path traversed by the algorithm.

\section{Numerical Experiments}\label{sec:examples}

To provide evidence for the performance of ITWF and comparisons to existing algorithms, we provide numerous examples in this section.


\subsection{Example 1: Sample complexity}
\begin{wrapfigure}{r}{0.5\textwidth}\vspace{-7mm}
\includegraphics[width=0.5\textwidth]{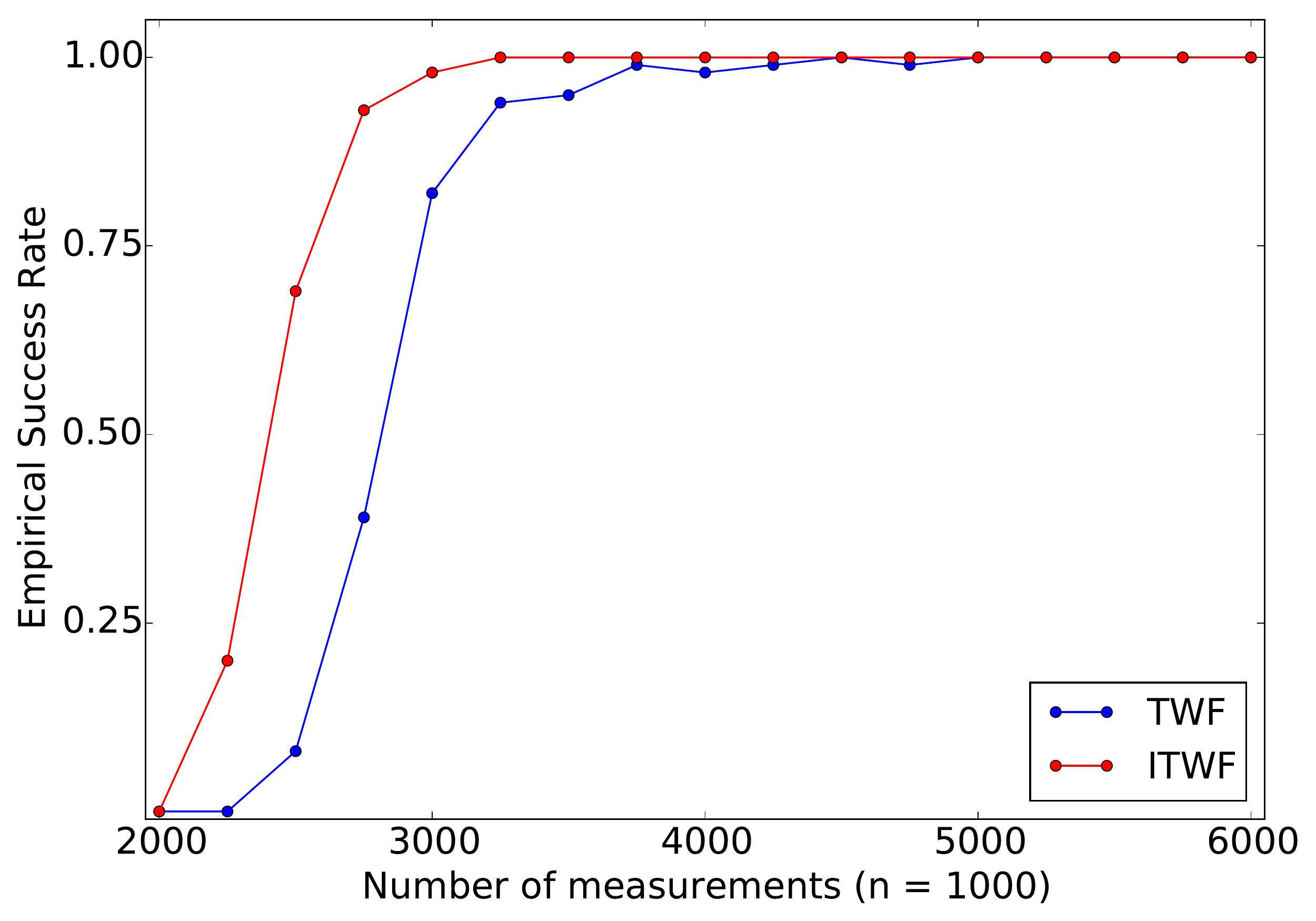}
\caption{Empirical Success Rate}\label{fig:example_1}
\end{wrapfigure} 
We compare the sample complexity of TWF and ITWF empirically when the sensing vectors $\bm{a}_i$ are i.i.d. $\mathcal{N}(0,I)$. The dimension $n$ is chosen to be 1000, and the number of measurements $m$ is varied from $2n$ to $6n$, and success is declared if the relative root mean squared error $\frac{\dist(\bm{z},\bm{x})}{\|\bm{x}\|}$ is less than $10^{-5}$ within 1000 passes through the data. The initialization uses 50 truncated power iterations. It can be seen from Figure~\ref{fig:example_1}, which is obtained by averaging over 100 Monte Carlo trials at each value of $m$, that the empirical success rate of ITWF is comparable with that of TWF, even slightly better.

\subsection{Example~2: Computational Complexity (Stage II)}
\begin{wrapfigure}{r}{0.5\textwidth}
\includegraphics[width=0.5\textwidth]{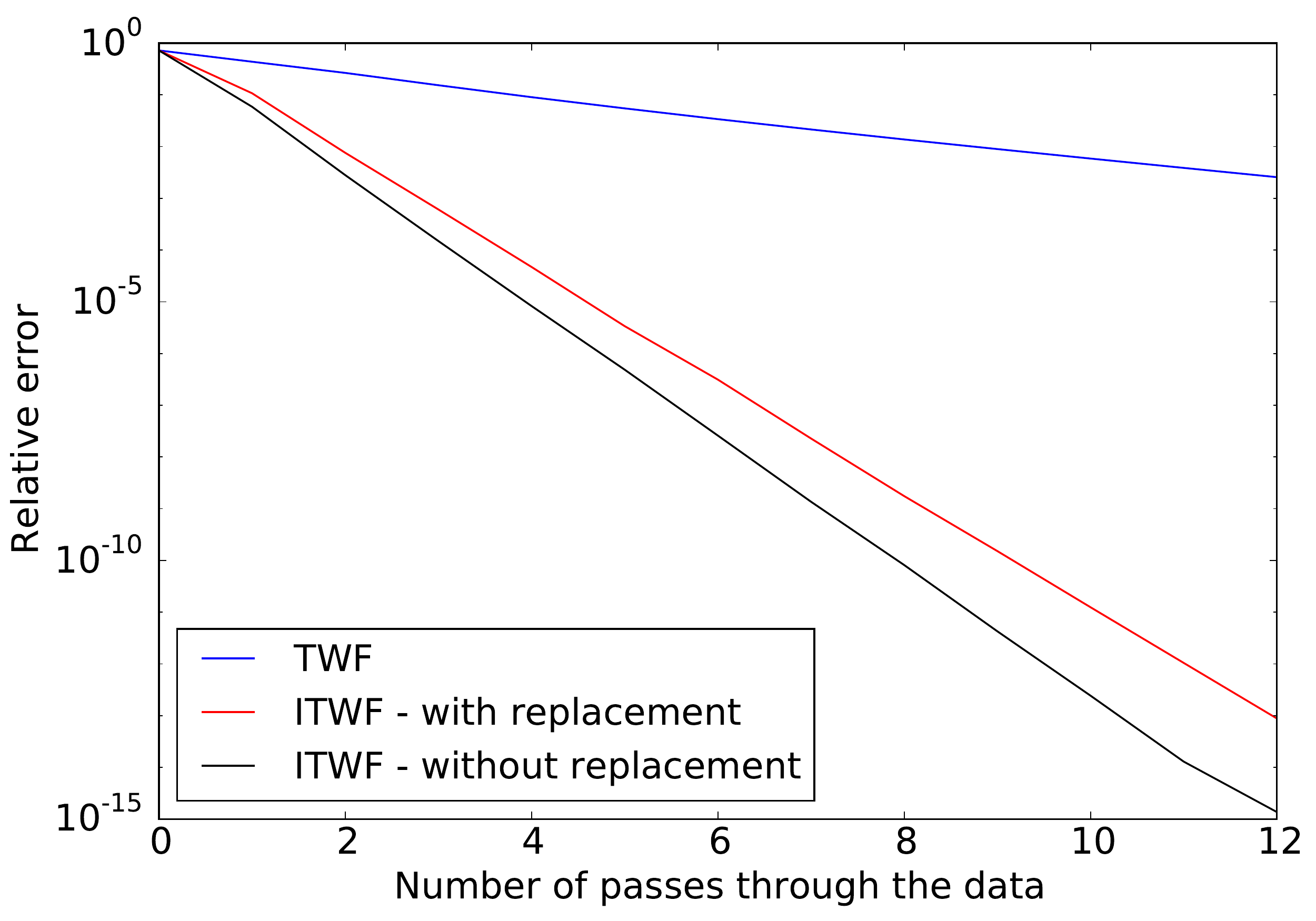}
\caption{Relative Error}\label{fig:example_2}
\end{wrapfigure} 
We consider the same setup as the previous example, with $m$ fixed to be 8000. Figure~\ref{fig:example_2} shows the relative root mean squared error $\frac{\dist(\bm{z},\bm{x})}{\|\bm{x}\|}$ of three algorithms, each run with the best step size, as a function of the number of passes through the data. All algorithms were initialized using 10 power iterations. 
The black line corresponds to the without-replacement variant of ITWF, in which all data points are visited exactly once in every block of $m$ iterations, each time in an independently chosen random order. The performance of the without-replacement variant is marginally better than the with-replacement variant, as has also been observed in numerous other contexts. However, the main point of the example is to show that ITWF offers substantial improvements in computational complexity over TWF. Owing to the better performance of the without-replacement variant, we adopt this sampling method for the remainder of the numerical experiments.

\subsection{Example 3: Structured Sensing Vectors}
To demonstrate the performance of the algorithm on a real signal when the assumption of random Gaussian sensing vectors do not hold, we consider another example from \cite{CheCan15}. An image of Stanford main quad of size 320 $\times$ 1280 pixels is used, and measurements are obtained via a set of $L$ coded diffraction patterns as
$$\bm{y}^{(l)} = (\bm{F}\bm{D}^{(l)}\bm{x})^* \odot (\bm{F}\bm{D}^{(l)}\bm{x}), \quad 1\leq l\leq L, $$
where $\bm{F}$ is the DFT matrix and $\bm{D}^{(l)}$ is a diagonal matrix (representing a mask) containing independent entries, each uniformly distributed over $\{+1, -1, j, -j\}$. The number of measurements is $nL$ with $L=12$, and the notation $|\cdot|^2$ denotes elementwise magnitude squared.

As in \cite{CheCan15}, we initialize the algorithm with 50 truncated power iterations, at the end of which the relative root mean squared error (rmse) is $0.35$. Another pass through the data using TWF gets it down to $0.15$. Increasing the number of passes to 2, 5 and 10 achieves $7.1\times 10^{-2}$, $2.3\times 10^{-2}$ and $6.6\times 10^{-3}$ respectively. Making one pass through the data using ITWF gets it down from $0.35$ to $8.2\times 10^{-3}$. Increasing the number of passes to 2, 5 and 10 achieves $3.0\times 10^{-4}$, $1.4\times 10^{-8}$ and $9.0\times 10^{-16}$ respectively.

\emph{Remark:} The sensing matrices in this example are highly structured (DFT, diagonal), due to which computing the sum of the $n$ gradients for one mask can be accomplished with $O(n\log n)$ computations instead of $O(n^2)$, by utilizing FFT algorithms. Hence, the motivation for ITWF that one iteration can be made $m$-times cheaper does not hold in this case. 
Hence we consider an \emph{increment} to be the set of measurements corresponding to 
one mask ($L$ measurements) instead of one measurement. This means that each iteration of ITWF is $L$-times cheaper than that of TWF where $L$ is the total number of masks; thus one iteration of TWF is computationally equivalent to $L$ iterations of ITWF. 

\begin{figure}[!t]
    \centering    
    \subfigure{\label{fig:ex_2_init}\includegraphics[width=\textwidth]{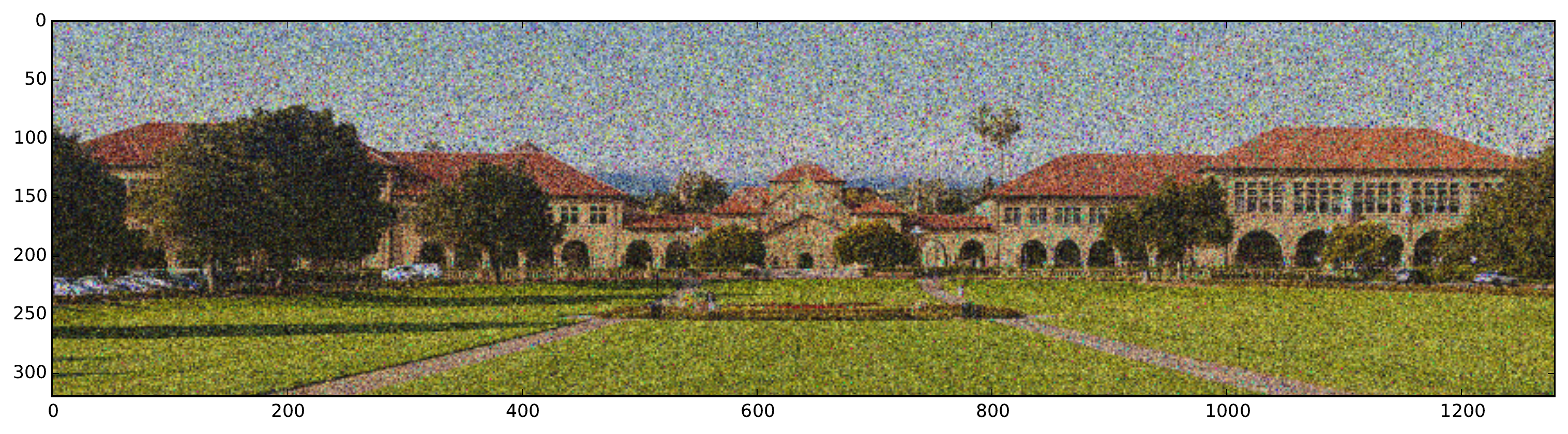}}\vspace{-2mm}
    \subfigure{\label{fig:ex_2_twf}\includegraphics[width=\textwidth]{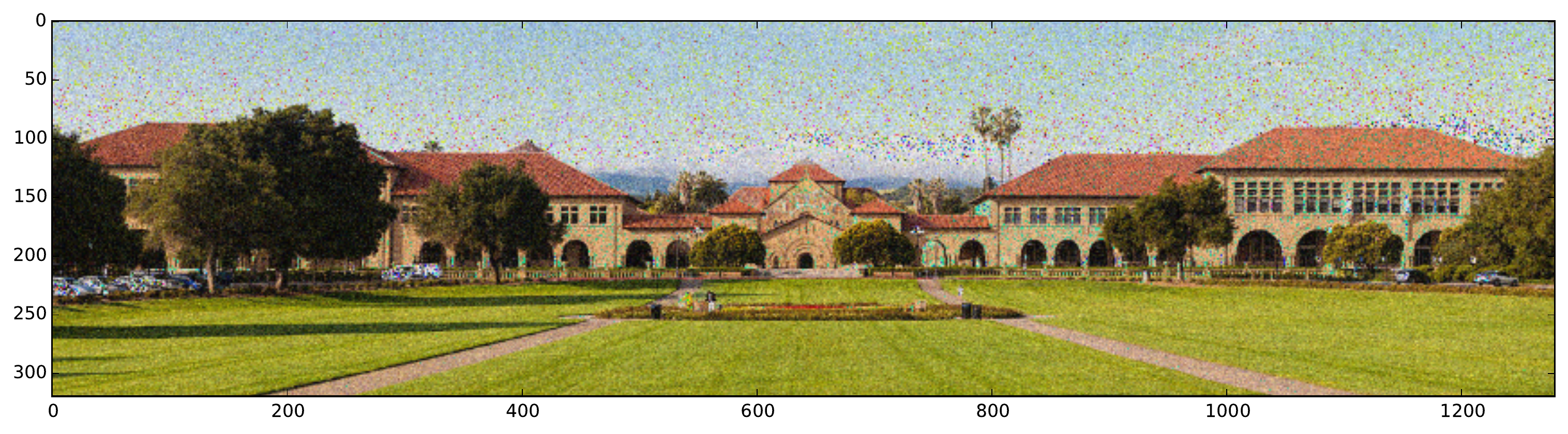}}\vspace{-2mm}
    \subfigure{\label{fig:ex_2_itwf}\includegraphics[width=\textwidth]{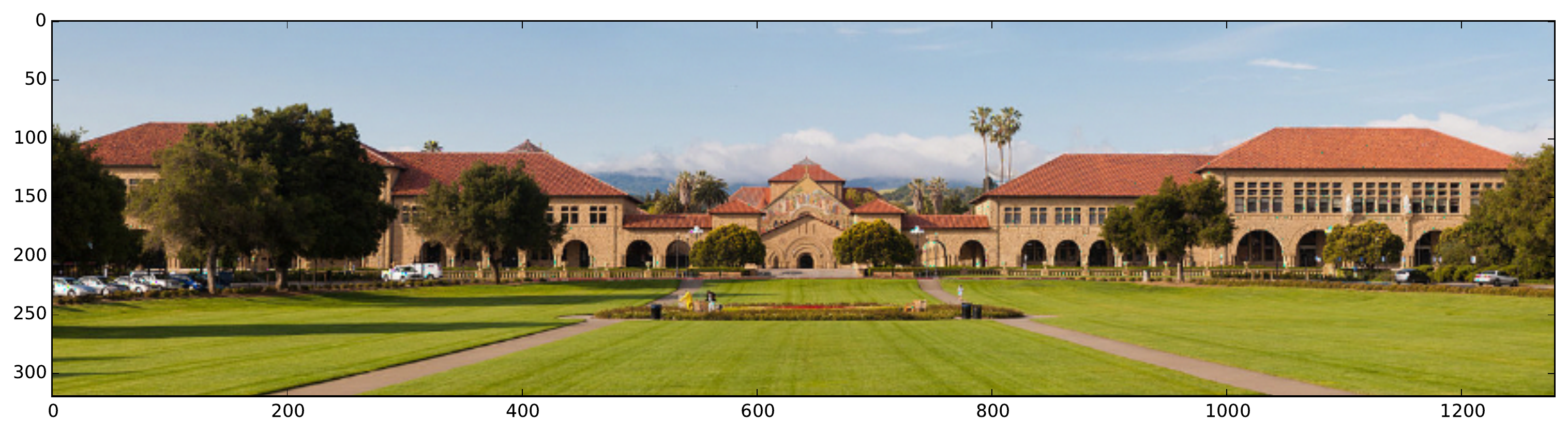}}    
        \caption{Results for Example 3: Recovered image after (top) initialization (50 power iterations), (center) 1 pass using TWF, (bottom) 1 pass using ITWF}
        \label{fig:example_3}
\end{figure}

\subsection{Example 4: Noisy measurements}
In this example, we generate noisy measurements as follows. As before, the sensing vectors are chosen to be i.i.d. $\mathcal{N}(0,I)$. The measurements are generated according to \eqref{eq:poisson}. Theorem~\ref{thm:2} effectively says that ITWF achieves \begin{equation}\label{eq:thm2_SNR}\frac{\dist^2(\bm{z}^{(T)},\bm{x})}{\|\bm{x}\|^2}\sim \frac{1}{\SNR},\end{equation} where $\SNR$ is the signal-to-noise ratio $\frac{\sum_{i=1}^m(\bm{a}_i^T\bm{x})^4}{\sum_{i=1}^m\eta_i^2}$, which is approximately $\frac{3m\|\bm{x}\|^4}{\|\bm{\eta}\|^2}.$

\begin{wrapfigure}{r}{0.5\textwidth}
\centering
\includegraphics[width=0.5\textwidth]{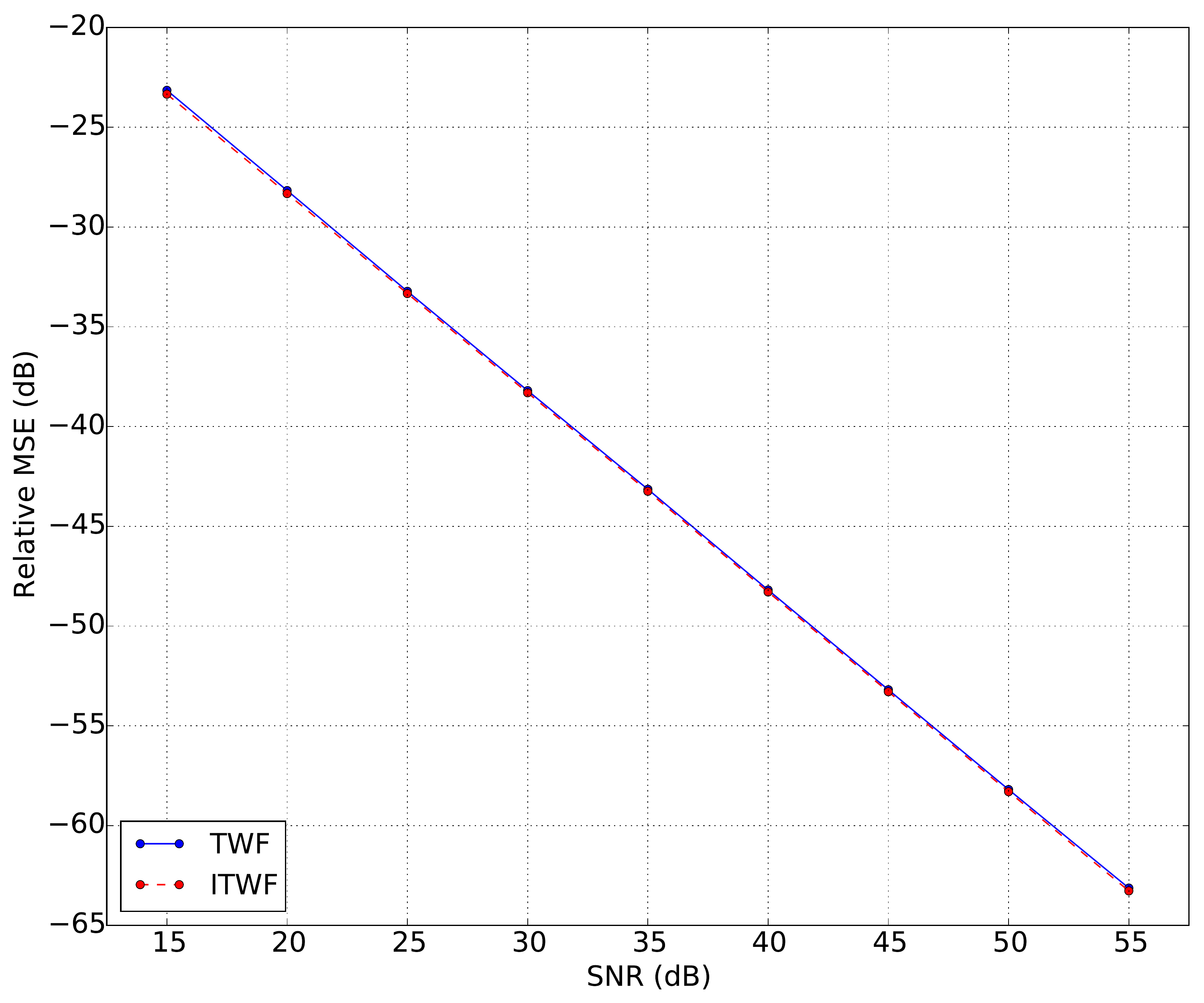}
\caption{Final Relative MSE achieved as a function of SNR}\label{fig:ex_4_1}
\end{wrapfigure}

Under the Poisson model, since $\sum_{i=1}^m\eta_i^2\approx \sum_{i=1}^m|\bm{a}_i^T\bm{x}|^2\approx m\|\bm{x}\|^2$, we refer to $3\|\bm{x}\|^2$ as the SNR. Comparing the final relative MSE of TWF and ITWF at various values of SNR Figure~\ref{fig:ex_4_1} shows that the final relative MSE (LHS of \eqref{eq:thm2_SNR}) of ITWF is in fact nearly equal to that of ITWF at all values of SNR.

The gain in computational complexity provided by ITWF that we observed in the noiseless case (ref. Figure~\ref{fig:example_2}) is however abated in the noisy case, as can be seen from the left panel in Figure~\ref{fig:example_4}. Experimenting with a diminishing step-size rule (step-size in $\ell$th pass chosen as $\Theta\left(\frac{1}{n\ell}\right)$), we find that it not only results in a small final MSE, but also reduces the number of passes it takes for ITWF to reach close to the final MSE by allowing us to exploit a larger step-size in the initial passes, as shown in the right panel in Figure~\ref{fig:example_4}.

\begin{figure*}[!h]
    \centering    
    \subfigure{\includegraphics[width=0.48\textwidth]{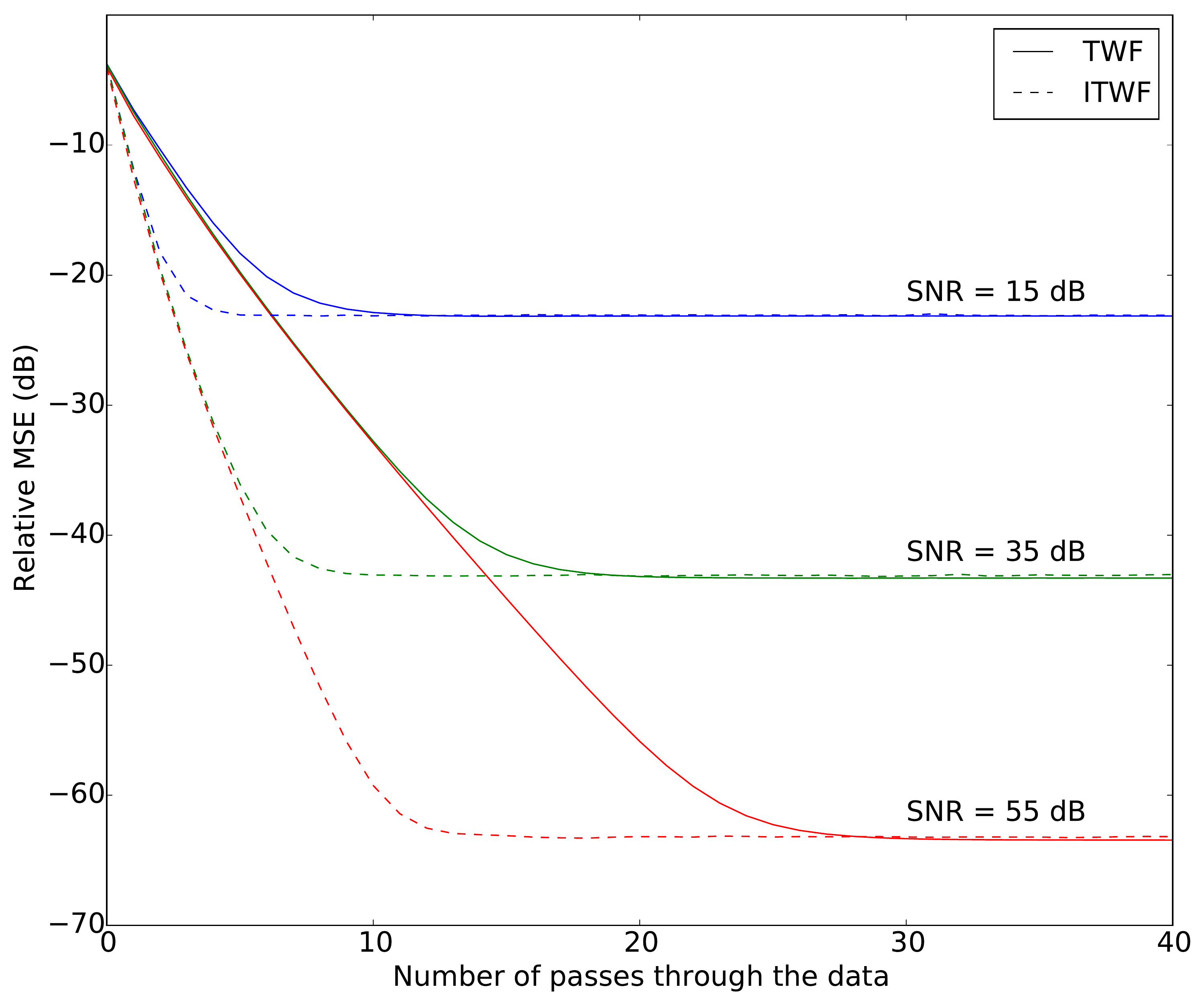}}\hspace{4mm}
        \subfigure{\includegraphics[width=0.48\textwidth]{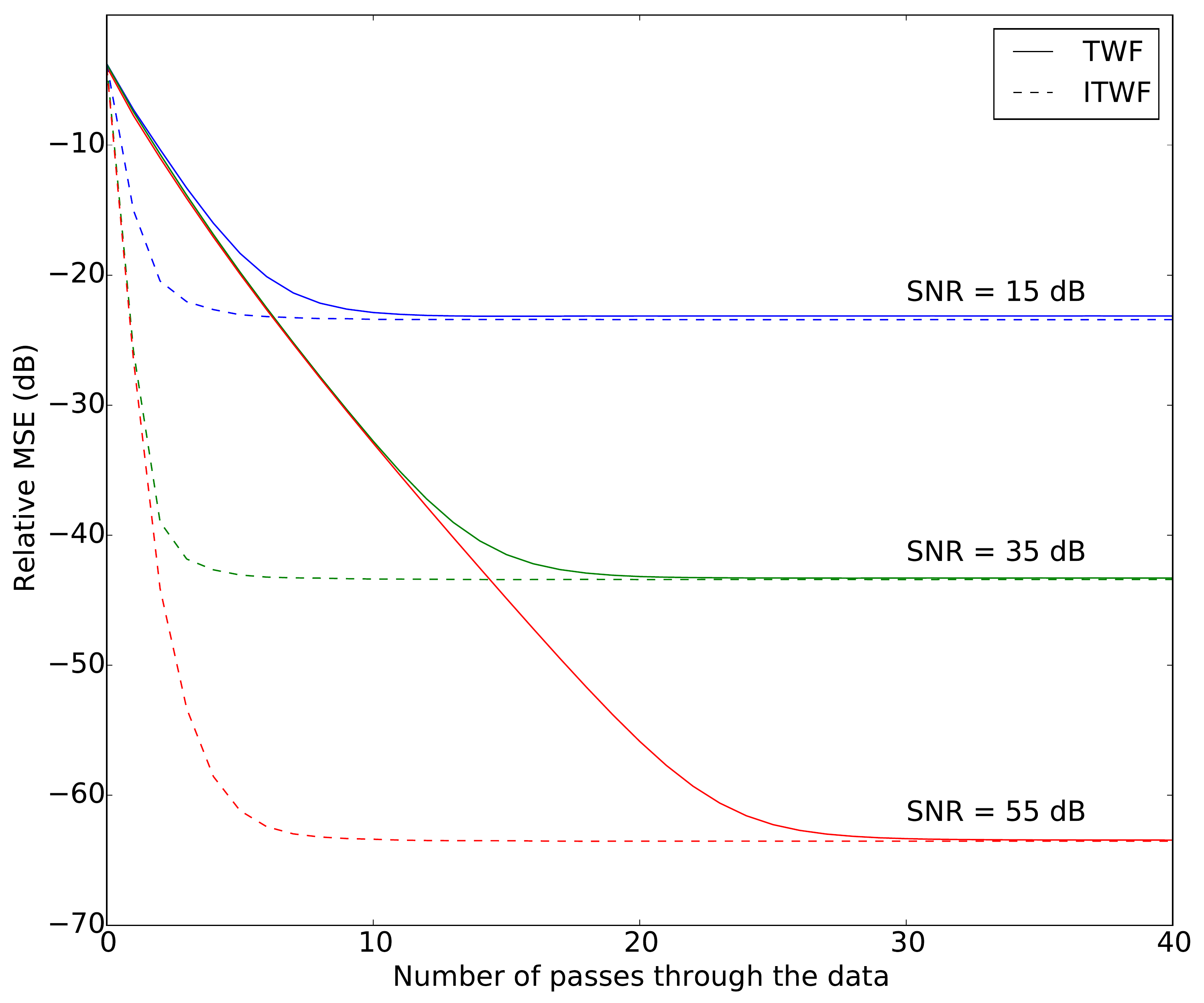}}
        \caption{Relative MSE achieved by TWF (step size 0.2) compared with (left) ITWF with constant step-size and (right) ITWF with a decreasing step size}\label{fig:example_4}
\end{figure*}

\section{Remarks and Future Directions}
While loosely bounding the constants arising during the analysis provides a recommendation for the step-size (for the random Gaussian model) which is $\approx \frac{0.00016}{n}$, this can be slightly pessimistic. In our numerical experiments, we find that larger step sizes also work.


Developing incremental methods that are able to extract the benefits of truncation for other problems such as sparse phase retrieval, low-rank matrix recovery from linear/quadratic measurements would be highly interesting.




\bibliography{phase_references}
\bibliographystyle{unsrt} 
\vfill
\pagebreak

\section{Appendix}

\subsection{Proof of Theorem~\ref{thm:1}}
The proof requires introducing some constants, and we collect their definitions before proceeding to the proof. 

Define
$$\zeta_1 := \max\left\{\begin{array}{c}\mathbb{E}\left[\xi^2\mathbbm{1}_{|\xi|\notin (\sqrt{1.01}\alpha_z^{\lb} , \sqrt{0.99}\alpha_z^{\ub})}\right], \\ \mathbb{E}\left[\mathbbm{1}_{|\xi|\notin (\sqrt{1.01}\alpha_z^{\lb} , \sqrt{0.99}\alpha_z^{\ub})}\right]\end{array} \right\},$$
and
$$\zeta_2 := \max\left\{\begin{array}{c}\mathbb{E}\left[\xi^2\mathbbm{1}_{|\xi|\geq  0.98\alpha_x }\right], \\ \mathbb{E}\left[\mathbbm{1}_{|\xi|\geq 0.98\alpha_x}\right]\end{array} \right\},$$
and
$$\zeta_3 := \mathbb{E}\left[\xi^2\mathbbm{1}_{|\xi|\geq  \sqrt{0.99}\gamma }\right],$$
where $\xi\sim\mathcal{N}(0,1)$ and $\gamma,$ $\alpha_x$, $\alpha_z^\ub$ and $\alpha_z^\lb$ can be chosen to be $5$, $0.3$, $5$ and $5$ respectively.

The term $\delta_{\init}$ appearing in the proofs can be considered to be equal to 0.1 for concreteness. 

The intermediate lemmas and propositions will establish statements that hold for all $\bm{z}$ that are in a neighborhood of $\bm{x}$ as follows:
\begin{equation}\label{eq:nbhd}
\frac{\dist(\bm{z},\bm{x})}{\|\bm{z}\|}\leq \min\left\{\delta_{\init}, \frac{\alpha_z^{\lb}}{6},\frac{\sqrt{98(\alpha_z^{\lb})^5}}{\sqrt{(\alpha_z^{\ub} + 1.01\alpha_x (1 + \delta_{\init}))^3}}  \right\},
\end{equation}
Note that the initialization stage can be used to guarantee that $\bm{z}^{(0)}$ is in this neighborhood of $\bm{x}$. 

Then, to prove the theorem, it suffices to prove Proposition~\ref{prop:1}. The reason why proving this proposition suffices is as follows. Consider the statement \eqref{eq:thm} in Proposition~\ref{prop:1}. We can take expectation of the RHS with respect to $i_{t-1}$, and then apply Proposition~\ref{prop:1} again to get a similar relation for the previous iteration. Continuing this process till we arrive at the initialization point $\bm{z}^{(0)}$, and noting that the initialization stage guarantees that $\dist(\bm{z}^{(0)}, \bm{x})$ is at most a small constant times $\|\bm{x}\|$, we arrive at the statement of Theorem~\ref{thm:1}. This concludes the proof of Theorem~\ref{thm:1}.\qed 

\begin{proposition}\label{prop:1}
When $\{\bm{a}_i\}_{i=1}^m$ are independent $\mathcal{N}(0,I)$, there exist universal constants $C,c_0,c_1,c_2,\rho > 0$, such that with probability at least $1 - Cm\exp(-c_1n)$ and $\mu = c_2/n$, we have
\begin{IEEEeqnarray}{l}
\mathbb{E}_{i_t}\left[\emph{dist}^2(\bm{z}^{(t+1)},\bm{x})\right] \leq \left(1 - \frac{\rho}{n}\right)\cdot\emph{dist}^2(\bm{z}^{(t)},\bm{x}),\IEEEeqnarraynumspace\label{eq:thm}
\end{IEEEeqnarray}
holds simultaneously for all $\bm{z}^{(t)}\in\mathbb{R}^n$ satisfying \eqref{eq:nbhd}, if $m \geq c_0n$.
\end{proposition}
\begin{proof}
Let $\bm{h}$ denote $\bm{z}^{(t)} - \bm{x}$, and $\mathcal{E}^{i}_t = \mathcal{E}^{i}_{1,t} \cap \mathcal{E}^{i}_{3}$ denote the overall truncation. 
The ITWF update \eqref{eq:inc_trunc_grad}, after some simple algebraic manipulations, gives us 
\begin{IEEEeqnarray}{rCl}
\mathbb{E}_{i_{t}}\left[\dist^2(\bm{z}^{(t+1)},\bm{x})\right] & = & \|\bm{h}\|^2 - \frac{4\mu}{m} \sum_{i=1}^m |\bm{a}_i^T\bm{h}|^2\left(1 + \frac{\bm{a}_i^T\bm{x}}{\bm{a}_i^T\bm{z}^{(t)}}\right)\mathbbm{1}_{\mathcal{E}^{i}_{t}}\nonumber\\
&& \quad +\> \frac{4\mu^2}{m}\sum_{i=1}^m \|\bm{a}_i\|^2|\bm{a}_i^T\bm{h}|^2\left(1 + \frac{\bm{a}_i^T\bm{x}}{\bm{a}_i^T\bm{z}^{(t)}}\right)^2 \mathbbm{1}_{\mathcal{E}^{i}_{t}}\label{eq:ITWF_noiseless_rec}
\end{IEEEeqnarray}

Lemmas 1, 2 and 3 show that for appropriate $\zeta_1$, $\zeta_2$ and $\zeta_3$ and any small $\delta>0$, the following two relations hold \emph{simultaneously} for all $\bm{z}$, $\bm{h}$ and $\bm{x}$ with high probability:
\begin{IEEEeqnarray*}{l}
\frac{1}{m}\sum_{i=1}^m |\bm{a}_i^T\bm{h}|^2\left(1 + \frac{\bm{a}_i^T\bm{x}}{\bm{a}_i^T\bm{z}}\right)\mathbbm{1}_{\mathcal{E}^{i}_{t}} \geq (0.99 - \zeta_1 - \zeta_2 - \zeta_3 - 3\delta)\|\bm{h}\|^2,
\end{IEEEeqnarray*}
and 
\begin{IEEEeqnarray*}{l}
\frac{1}{m}\sum_{i=1}^m \|\bm{a}_i\|^2|\bm{a}_i^T\bm{h}|^2\left(1 + \frac{\bm{a}_i^T\bm{x}}{\bm{a}_i^T\bm{z}}\right)^2 \mathbbm{1}_{\mathcal{E}^{i}_{t}} \leq c_4 n (1+\delta) \|\bm{h}\|^2,
\end{IEEEeqnarray*}
if $m$ is $O(n)$ and $c_4$ is a constant. The proofs mainly involve the application of concentration bounds and $\epsilon$-net arguments. Substituting these relations in \eqref{eq:ITWF_noiseless_rec} and choosing $\mu$ to be say $\frac{0.99 - \zeta_1 - \zeta_2 - \zeta_3 - 3\delta}{2c_4(1+\delta)n}$, we get that 
$$\mathbb{E}_{i_{t}}\left[\dist^2(\bm{z}^{(t+1)},\bm{x})\right] \leq \left(1 - \frac{\rho}{n}\right)\dist^2(\bm{z}^{(t)},\bm{x}).$$
\end{proof}

\begin{lemma}\label{lem:1}
For any $\delta > 0$, there exist universal constants $C, c_0, c_1>0$ such that if $m > c_1n\delta^{-2}\log (1/\delta)$, then
$$\frac{1}{m} \sum_{i=1}^m |\bm{a}_i^T\bm{h}|^2\mathbbm{1}_{\mathcal{E}^{i}_{1,t}} \mathbbm{1}_{\mathcal{E}^{i}_{3}}  \geq (1 - \zeta_1 - \zeta_2 - \zeta_3 - 2\delta )\|\bm{h}\|^2,$$
with probability $1 - C\exp(-c_0m\delta^2)$, simultaneously for all non-zero vectors $\bm{h},\bm{z},\bm{x}\in\mathbb{R}^n$.
\end{lemma}
\begin{proof}
The proof follows a similar path as Lemma 4 in \cite{CheCan15}. We can restrict attention to unit norm $\bm{h}$, $\bm{z}$ and $\bm{x}$. We first assume that $\bm{z}$, $\bm{h}$ and $\bm{x}$ are independent from $\{\bm{a}_i\}_i$. Later, by using an $\epsilon$-net argument, we extend it to all unit norm vectors $\bm{z}$, $\bm{h}$ and $\bm{x}$. 

First, we note that with high probability, we have
\begin{IEEEeqnarray*}{rCl}
\mathbbm{1}_{\mathcal{E}^{i}_{3}} = \mathbbm{1}_{\left\{ \frac{y_i}{\frac{1}{m}\sum_{i=1}^my_i } \leq \alpha^2_x\right\}} & \geq & \mathbbm{1}_{\left\{ \frac{|\bm{a}_i^T\bm{x}|}{\|\bm{x}\| }\leq \alpha_x (1-0.01) \right\}}\\
& = & \mathbbm{1}_{ \left\{\frac{|\bm{a}_i^T\bm{x}|}{\|\bm{x}\| }\leq 0.99\alpha_x \right\} },
\end{IEEEeqnarray*}
where the  inequality follows from Theorem 5.39 in \cite{Ver10}, which says that \begin{equation}\label{eq:ver539}\left\|\frac{1}{m}\sum_{i=1}^m \bm{a}_i\bm{a}_i^T - I\right\| \leq \delta,\end{equation}
with probability $1 - C\exp(-c_1m\delta^2)$ if $m\geq c_0n\delta^{-2}$. For convenience, let us denote $0.99\alpha_x$ by $\widetilde{\alpha}_x$. In the remainder of this proof and the rest of the proofs for the noiseless case, we will use the established lower bound
$$\mathbbm{1}_{\mathcal{E}^{i}_{3}} \geq \mathbbm{1}_{ \left\{\frac{|\bm{a}_i^T\bm{x}|}{\|\bm{x}\| }\leq \widetilde{\alpha}_x \right\} }$$
implicitly whenever required.

Now define the functions
\begin{IEEEeqnarray*}{c}
\chi_z(\tau) = \begin{cases}1, & \text{ if } |\tau| \in [ \sqrt{1.01}\alpha_z^{\lb}, \sqrt{0.99}\alpha_z^{\ub} ] \\
100\left(1 - \left(\frac{\tau}{\alpha_z^{\ub}}\right)^2 \right), & \text{ if }|\tau| \in [ \sqrt{0.99}\alpha_z^{\ub}, \alpha_z^{\ub} ]\\
100\left(\left(\frac{\tau}{\alpha_z^{\lb}}\right)^2 -1\right), & \text{ if }|\tau| \in [\alpha_z^{\lb}, \sqrt{1.01}\alpha_z^{\lb} ]\\
0, & \text{ else,} \end{cases} 
\end{IEEEeqnarray*}
and
\begin{IEEEeqnarray*}{c}
\chi_x(\tau) = \begin{cases}1, & \text{ if } |\tau| \leq \sqrt{0.99}\widetilde{\alpha}_x \\
100\left(1 - \left(\frac{\tau}{\widetilde{\alpha}_x}\right)^2 \right), & \text{ if }|\tau| \in [ \sqrt{0.99}\widetilde{\alpha}_x, \widetilde{\alpha}_x ]\\
0, & \text{ else,} \end{cases} 
\end{IEEEeqnarray*}
and
\begin{IEEEeqnarray*}{c}
\chi_h(\tau) = \begin{cases}1, & \text{ if } |\tau| \leq \sqrt{0.99}\gamma \\
100\left(1 - \left(\frac{\tau}{\gamma}\right)^2 \right), & \text{ if }|\tau| \in [ \sqrt{0.99}\gamma, \gamma ]\\
0, & \text{ else.} \end{cases} 
\end{IEEEeqnarray*}

Since $\bm{h}$, $\bm{x}$ and $\bm{z}$ are assumed to be unit vectors, we have $$  \mathbbm{1}_{\mathcal{E}^{i}_{1,t}} \mathbbm{1}_{\mathcal{E}^{i}_{3}} \geq\mathbbm{1}_{\mathcal{E}^{i}_{1,t}} \mathbbm{1}_{\mathcal{E}^{i}_{3}} \mathbbm{1}_{\mathcal{E}^{i}_{4,t}} \geq \chi_z(\bm{a}_i^T\bm{z})\chi_x(\bm{a}_i^T\bm{x})\chi_h(\bm{a}_i^T\bm{h}) \geq 0,$$ 
where 
$\mathcal{E}^{i}_{4,t}$ is defined to be the event $\{|\bm{a}_i^T\bm{h}| \leq \gamma\|\bm{h}\|\}$, and $\gamma$ can be chosen to be 5. This means that 
\begin{equation}\label{eq:lem1_1}\frac{1}{m} \sum_{i=1}^m |\bm{a}_i^T\bm{h}|^2\mathbbm{1}_{\mathcal{E}^{i}_{1,t}} \mathbbm{1}_{\mathcal{E}^{i}_{3}}\geq \frac{1}{m} \sum_{i=1}^m |\bm{a}_i^T\bm{h}|^2\chi_z(\bm{a}_i^T\bm{z})\chi_x(\bm{a}_i^T\bm{x})\chi_h(\bm{a}_i^T\bm{h}).\end{equation}
The reason for introducing these functions is that their $O(1)$-Lipschitz continuity allows us to apply the $\epsilon$-net argument later in the proof. The event $\mathcal{E}^{i}_{4,t}$ allows us to discard atypically large terms in the summation, which is useful while applying the $\epsilon$-net argument.

We now lower bound the RHS of the above inequality. The expectation of each term in the summation can be lower bounded as follows:
\begin{IEEEeqnarray*}{l}
\mathbb{E}\left[|\bm{a}_i^T\bm{h}|^2\chi_z(\bm{a}_i^T\bm{z})\chi_x(\bm{a}_i^T\bm{x})\chi_h(\bm{a}_i^T\bm{h})\right]\\ 
\quad \geq  \mathbb{E}\left[|\bm{a}_i^T\bm{h}|^2\chi_z(\bm{a}_i^T\bm{z})\right] + \mathbb{E}\left[|\bm{a}_i^T\bm{h}|^2\chi_x(\bm{a}_i^T\bm{x})\right] + \mathbb{E}\left[|\bm{a}_i^T\bm{h}|^2\chi_h(\bm{a}_i^T\bm{h})\right] - 2\mathbb{E}\left[|\bm{a}_i^T\bm{h}|^2\right]\\
\quad \geq (1 - \zeta_1)\|\bm{h}\|^2 + (1 - \zeta_2)\|\bm{h}\|^2 + (1 - \zeta_3)\|\bm{h}\|^2 - 2\|\bm{h}\|^2 \\ 
\quad = (1 - \zeta_1 - \zeta_2 - \zeta_3)\|\bm{h}\|^2,
\end{IEEEeqnarray*}
where these inequalities follow by similar arguments as equation (120) in \cite{CheCan15}.

Observing that $|\bm{a}_i^T\bm{h}|^2\chi_z(\bm{a}_i^T\bm{z})\chi_x(\bm{a}_i^T\bm{x})\chi_h(\bm{a}_i^T\bm{h})$ is a sub-exponential random variable with sub-exponential norm $O(\|\bm{h}\|^2)$, we can use Proposition 5.16 from \cite{Ver10} to get that with probability $1 - \exp(-\Omega(\delta^2 m))$, 
$$\frac{1}{m} \sum_{i=1}^m |\bm{a}_i^T\bm{h}|^2\chi_z(\bm{a}_i^T\bm{z})\chi_x(\bm{a}_i^T\bm{x})\chi_h(\bm{a}_i^T\bm{h}) \geq (1 - \zeta_1 - \zeta_2 - \zeta_3 - \delta)\|\bm{h}\|^2.$$

We now construct an $\epsilon$-net $\mathcal{N}_\epsilon$, such that for any $(\bm{h},\bm{z},\bm{x})$ with $\|\bm{h}\| = \|\bm{z}\| = \|\bm{x}\| = 1$, there exists $(\bm{h}_0,\bm{z}_0,\bm{x}_0)\in\mathcal{N}_{\epsilon}$ such that $\|\bm{h}-\bm{h}_0\|\leq \epsilon$, $\|\bm{z}-\bm{z}_0\|\leq \epsilon$ and $\|\bm{x}-\bm{x}_0\|\leq \epsilon$, and $|\mathcal{N}_\epsilon| \leq \left(1 + \frac{2}{\epsilon}\right)^{3n}.$
Taking a union bound over this set gives us that
$$\frac{1}{m} \sum_{i=1}^m |\bm{a}_i^T\bm{h}_0|^2\chi_z(\bm{a}_i^T\bm{z}_0)\chi_x(\bm{a}_i^T\bm{x}_0)\chi_h(\bm{a}_i^T\bm{h}_0) \geq (1 - \zeta_1 - \zeta_2 - \zeta_3 - \delta)\|\bm{h}_0\|^2,$$
holds for all $(\bm{h}_0,\bm{z}_0,\bm{x}_0)\in\mathcal{N}_\epsilon$ with probability a least $1 - \left(1 + \frac{2}{\epsilon}\right)^{3n}\exp(-\Omega(\delta^2 m))$.

Now, since $\chi_z(\sqrt{\tau})$, $\chi_x(\sqrt{\tau})$ and $\tau\chi_h(\sqrt{\tau})$ are Lipschitz functions with Lipschitz constant $O(1)$, we get that for any unit vector pair $\bm{h}$, $\bm{z}$ and $\bm{x}$, 
\begin{IEEEeqnarray*}{l}
\left| |\bm{a}_i^T\bm{h}|^2\chi_z(\bm{a}_i^T\bm{z})\chi_x(\bm{a}_i^T\bm{x})\chi_h(\bm{a}_i^T\bm{h}) - |\bm{a}_i^T\bm{h}_0|^2\chi_z(\bm{a}_i^T\bm{z}_0)\chi_x(\bm{a}_i^T\bm{x}_0)\chi_h(\bm{a}_i^T\bm{h}_0)\right| \\
\quad \leq |\chi_z(\bm{a}_i^T\bm{z}_0)\chi_x(\bm{a}_i^T\bm{x}_0)|\cdot\left||\bm{a}_i^T\bm{h}|^2\chi_h(\bm{a}_i^T\bm{h}) - |\bm{a}_i^T\bm{h}_0|^2\chi_h(\bm{a}_i^T\bm{h}_0)\right|\\
\quad\quad +\> |\bm{a}_i^T\bm{h}_0|^2|\chi_h(\bm{a}_i^T\bm{h}_0)\chi_x(\bm{a}_i^T\bm{x}_0)|\cdot \left|\chi_z(\bm{a}_i^T\bm{z}) - \chi_z(\bm{a}_i^T\bm{z}_0)\right|\\
\quad\quad +\> |\bm{a}_i^T\bm{h}_0|^2\chi_h(\bm{a}_i^T\bm{h}_0)|\chi_z(\bm{a}_i^T\bm{z}_0)|\cdot \left|\chi_x(\bm{a}_i^T\bm{x}) - \chi_x(\bm{a}_i^T\bm{x}_0)\right|\\
\quad\lesssim \left||\bm{a}_i^T\bm{h}|^2 - |\bm{a}_i^T\bm{h}_0|^2\right| + \left||\bm{a}_i^T\bm{z}|^2 - |\bm{a}_i^T\bm{z}_0|^2\right| + \left||\bm{a}_i^T\bm{x}|^2 - |\bm{a}_i^T\bm{x}_0|^2\right|.
\end{IEEEeqnarray*}

Hence, there exists a universal constant $c_3$ such that
\begin{IEEEeqnarray*}{l}
\left|\frac{1}{m} \sum_{i=1}^m |\bm{a}_i^T\bm{h}|^2\chi_z(\bm{a}_i^T\bm{z})\chi_x(\bm{a}_i^T\bm{x}) - \frac{1}{m} \sum_{i=1}^m |\bm{a}_i^T\bm{h}_0|^2\chi_z(\bm{a}_i^T\bm{z}_0)\chi_x(\bm{a}_i^T\bm{x}_0) \right| \\
\quad\lesssim \frac{1}{m}\sum_{i=1}^m\left(\left||\bm{a}_i^T\bm{h}|^2 - |\bm{a}_i^T\bm{h}_0|^2\right| + \left||\bm{a}_i^T\bm{z}|^2 - |\bm{a}_i^T\bm{z}_0|^2\right|+ \left||\bm{a}_i^T\bm{x}|^2 - |\bm{a}_i^T\bm{x}_0|^2\right|\right)\\
\quad\leq 3c_3 \epsilon,
\end{IEEEeqnarray*}
where the last inequality follows by invoking Lemma 1 and Lemma 2 from \cite{CheCan15}, assuming $\epsilon < 1/2$. By choosing $\epsilon = \frac{\delta}{3c_3}$ and recalling \eqref{eq:lem1_1}, we get the lemma.
\end{proof}

\begin{lemma}\label{lem:2}
Consider any $\bm{z}\in\mathbb{R}^n$, and $\bm{h}:=\bm{z}-\bm{x}$ such that
$$\frac{\|\bm{h}\|}{\|\bm{z\|}} \leq \min\left\{\delta_{\emph{init} }, \frac{\alpha_z^{\emph{lb}}}{6},\frac{\sqrt{98(\alpha_z^{\emph{lb}})^5}}{\sqrt{(\alpha_z^{\ub} + 1.01\alpha_x (1 + \delta_{\init}))^3}}  \right\}.$$
Then, for any $\delta \in (0,1)$, there exist universal constants $C,c_0,c_1>0$ such that if $m\geq c_1n\delta^{-2}\log(1/\delta)$, it holds that
$$\frac{1}{m}\sum_{i=1}^m |\bm{a}_i^T\bm{h}|^2 \frac{\bm{a}_i^T\bm{x}}{\bm{a}_i^T\bm{z}}\mathbbm{1}_{\mathcal{E}^{i}_{t}} \geq -(0.01 + \delta)\|\bm{h}\|^2,$$
with probability at least $1 - C\exp(-c_0\delta^2m)$.
\end{lemma}
\begin{proof}
We first note that
\begin{IEEEeqnarray*}{l}
\frac{1}{m}\sum_{i=1}^m |\bm{a}_i^T\bm{h}|^2 \frac{\bm{a}_i^T\bm{x}}{\bm{a}_i^T\bm{z}}\mathbbm{1}_{\mathcal{E}^{i}_{t}}\\
\quad\quad\quad\quad \geq \frac{1}{m}\sum_{i=1}^m |\bm{a}_i^T\bm{h}|^2 \frac{\bm{a}_i^T\bm{x}}{\bm{a}_i^T\bm{z}}\mathbbm{1}_{\mathcal{E}^{i}_{t}\cap \mathcal{D}^i_t}\\
\quad\quad\quad\quad =  -\left(\frac{1}{m}\sum_{i=1}^m |\bm{a}_i^T\bm{h}|^2 \left|\frac{\bm{a}_i^T\bm{x}}{\bm{a}_i^T\bm{z}}\right|\mathbbm{1}_{\mathcal{E}^{i}_{t}\cap \mathcal{D}^i_t}\right),
\end{IEEEeqnarray*}
where 
$$\mathcal{D}^i_t = \left\{\frac{\bm{a}_i^T\bm{x}}{\bm{a}_i^T\bm{z}} \leq 0\right\}.$$
Also note that $\mathcal{D}^i_t \Leftrightarrow \widetilde{\mathcal{D}}^i_t,$
where 
$$\widetilde{\mathcal{D}}^i_t = \left\{|\bm{a}_i^T\bm{h}| \geq |\bm{a}_i^T\bm{z}|\right\} \cap \left\{|\bm{a}_i^T\bm{h}| \geq |\bm{a}_i^T\bm{x}|\right\}. $$

So, we now proceed to obtain an upper bound on $\frac{1}{m}\sum_{i=1}^m |\bm{a}_i^T\bm{h}|^2 \left|\frac{\bm{a}_i^T\bm{x}}{\bm{a}_i^T\bm{z}}\right|\mathbbm{1}_{\mathcal{E}^{i}_{t}\cap \widetilde{\mathcal{D}}^i_t}$. We have
\begin{IEEEeqnarray*}{l}
\frac{1}{m}\sum_{i=1}^m |\bm{a}_i^T\bm{h}|^2 \left|\frac{\bm{a}_i^T\bm{x}}{\bm{a}_i^T\bm{z}}\right|\mathbbm{1}_{\mathcal{E}^{i}_{t}\cap \widetilde{\mathcal{D}}^i_t}\\
\quad\quad\quad\quad\leq \frac{1}{m}\sum_{i=1}^m \frac{|\bm{a}_i^T\bm{h}|^3}{|\bm{a}_i^T\bm{z}|} \mathbbm{1}_{\mathcal{E}^{i}_{t}\cap \widetilde{\mathcal{D}}^i_t}\\
\quad\quad\quad\quad\leq \frac{1}{m}\sum_{i=1}^m \mathbbm{1}_{\mathcal{E}^{i}_{t}\cap \widetilde{\mathcal{D}}^i_t}\max_{i:\mathcal{E}^{i}_{t}\cap \widetilde{\mathcal{D}}^i_t}\left\{\frac{|\bm{a}_i^T\bm{h}|^3}{|\bm{a}_i^T\bm{z}|}\right\}.
\end{IEEEeqnarray*}

Now, we can use Lemma 6 from \cite{CheCan15} to upper bound $\sum_{i=1}^m \mathbbm{1}_{\mathcal{E}^{i}_{t}\cap \widetilde{\mathcal{D}}^i_t}$ with high probability, as long as $m$ is sufficiently large. Assuming that $\gamma$ is a real number that is at least 2 and at most $\frac{\alpha_z^{\lb}\|\bm{z}\|}{\|\bm{h}\|}$, this is done as follows:
\begin{IEEEeqnarray*}{rCl}
\frac{1}{m}\sum_{i=1}^m \mathbbm{1}_{\mathcal{E}^{i}_{t}\cap \widetilde{\mathcal{D}}^i_t} & \leq & \frac{1}{m}\sum_{i=1}^m \mathbbm{1}_{\{|\bm{a}_i^T\bm{h}|\geq \gamma\|\bm{h}\|\}}\\
& \leq & \frac{1}{0.49\gamma}\exp(-0.485\gamma^2) + \frac{\epsilon}{\gamma^2}\\
& \leq & \frac{1}{9800}\left(\frac{\|\bm{h}\|}{\alpha_z^{\lb}\|\bm{z}\|}\right)^4 + \frac{\epsilon}{(\alpha_z^{\lb})^2}\left(\frac{\|\bm{h}\|}{\|\bm{z}\|}\right)^2,
\end{IEEEeqnarray*}
which holds simultaneously for all $\bm{h}$ with probability $1 - C\exp(-c_0\epsilon^2 m)$ if $m\geq c_1n\epsilon^{-2}\log(1/\epsilon)$, where the last step also requires $\frac{\|\bm{h}\|}{\|\bm{z}\|}\leq \frac{\alpha_z^{\lb}}{6}$.

Since $|\bm{a}_i^T\bm{h}| = |\bm{a}_i^T\bm{z} - \bm{a}_i^T\bm{x}| \leq |\bm{a}_i^T\bm{z}| + |\bm{a}_i^T\bm{x}|$ and by the definition of $\mathcal{E}^{i}_{t}$, we have
\begin{IEEEeqnarray*}{rCl}
\max\limits_{i:\mathcal{E}^{i}_{t}\cap \widetilde{\mathcal{D}}^i_t}\left\{\frac{|\bm{a}_i^T\bm{h}|^3}{|\bm{a}_i^T\bm{z}|}\right\} & \leq & \frac{(\alpha_z^{\ub} + 1.01\alpha_x (1 + \delta_{\init}))^3}{\alpha_z^{\lb}}\|\bm{z}\|^2
\end{IEEEeqnarray*}
if we assume $\|\bm{x}\|\leq (1+\delta_{\init})\|\bm{z}\|$, which is true if $\frac{\|\bm{h}\|}{\|\bm{z}\|}\leq \delta_{\init}$. 

As a result of these two upper bounds, we get 
\begin{IEEEeqnarray*}{l}
\frac{1}{m}\sum_{i=1}^m \mathbbm{1}_{\mathcal{E}^{i}_{t}\cap \widetilde{\mathcal{D}}^i_t}\max_{i:\mathcal{E}^{i}_{t}\cap \widetilde{\mathcal{D}}^i_t}\left\{\frac{|\bm{a}_i^T\bm{h}|^3}{|\bm{a}_i^T\bm{z}|}\right\}\\
\quad \leq \left\{ \frac{(\alpha_z^{\ub} + 1.01\alpha_x (1 + \delta_{\init}))^3}{9800(\alpha_z^{\lb})^5}\frac{\|\bm{h}\|^2}{\|\bm{z}\|^2} + \frac{\epsilon(\alpha_z^{\ub} + 1.01\alpha_x (1 + \delta_{\init}))^3}{(\alpha_z^{\lb})^3} \right\}\|\bm{h}\|^2.
\end{IEEEeqnarray*}
Now, by choosing $\epsilon$ appropriately, and if $$\frac{\|\bm{h}\|}{\|\bm{z}\|}\leq \frac{\sqrt{98(\alpha_z^{\lb})^5}}{\sqrt{(\alpha_z^{\ub} + 1.01\alpha_x (1 + \delta_{\init}))^3}},$$
we get the claim of the lemma.
\end{proof}

\begin{lemma}\label{lem:3}
For all non-zero vectors $\bm{z}\in\mathbb{R}^n$ and $\bm{h}:=\bm{z} - \bm{x}$ with $\|\bm{h}\|\leq \delta_{\emph{init}}\|\bm{z}\|$, we have if $m\geq c_1n\delta^{-2}$,
$$\frac{1}{m}\sum_{i=1}^m \|\bm{a}_i\|^2 |\bm{a}_i^T\bm{h}|^2 \left(1 + \frac{\bm{a}_i^T\bm{x}}{\bm{a}_i^T\bm{z}}\right)^2\mathbbm{1}_{\mathcal{E}^{i}_{t}} \leq c_4(1 + \delta)n\|\bm{h}\|^2,$$
with probability $1 - Cm\exp(-c_0n\delta^2)$.  
\end{lemma}
\begin{proof}
We have that $\|\bm{a}_i\| \leq \sqrt{6n}$ for all $i$ with probability $1 - m\exp(-1.5n)$. The following chain of inequalities uses this, and the definition of $\mathcal{E}^i_t.$
\begin{IEEEeqnarray*}{l}
\frac{1}{m}\sum_{i=1}^m \|\bm{a}_i\|^2 |\bm{a}_i^T\bm{h}|^2 \left(1 + \frac{\bm{a}_i^T\bm{x}}{\bm{a}_i^T\bm{z}}\right)^2\mathbbm{1}_{\mathcal{E}^{i}_{t}}\\
\quad \leq \frac{1}{m}\sum_{i=1}^m \|\bm{a}_i\|^2 |\bm{a}_i^T\bm{h}|^2 \left(1 + \frac{1.01\alpha_x(1+ \delta_{\init})}{\alpha_z^{\lb}}\right)^2\mathbbm{1}_{\mathcal{E}^{i}_{t}}\\
\quad \leq 6n\left(1 + 1.01\frac{\alpha_x(1+ \delta_{\init})}{\alpha_z^{\lb}}\right)^2\frac{1}{m} \sum_{i=1}^m|\bm{a}_i^T\bm{h}|^2 \\
\quad\leq c_4n(1+\delta)\|\bm{h}\|^2,
\end{IEEEeqnarray*}
where the last inequality follows since by Theorem 5.39 in \cite{Ver10}, $$\left\|\frac{1}{m}\sum_{i=1}^m \bm{a}_i\bm{a}_i^T - I\right\| \leq \delta,$$
with probability $1 - C\exp(-c_1m\delta^2)$ if $m\geq c_0n\delta^{-2}$. We have also used the fact that $y_i \leq \alpha^2_x (\frac{1}{m}\sum_{i=1}^my_i)$ implies $$ |\bm{a}_i^T\bm{x}| \leq 1.01\alpha_x\|\bm{x}\|\leq 1.01\alpha_x(1+\delta_{\init})\|\bm{z}\|.$$ It can be seen that $c_4$ stands for the quantity $6\left(1 + \frac{1.01\alpha_x(1+ \delta_{\init})}{\alpha_z^{\lb}}\right)^2$.
\end{proof}
\pagebreak
\subsection{Proof of Theorem~\ref{thm:2} }

In the noisy case, the measurements are generated according to $y_i = |\bm{a}_i^T\bm{x}|^2 + \eta_i$. Recall that we are using $\mathcal{E}^{i}_t$ to denote $\mathcal{E}^{i}_{1,t} \cap \mathcal{E}^{i}_{3}$.

Proceeding in a similar manner as Proposition~\ref{prop:1}, and recalling that we are using $\bm{h} = \bm{z}^{(t)} - \bm{x}$, we get that
\begin{IEEEeqnarray}{l}
\mathbb{E}_{i_{t}}\left[\dist^2(\bm{z}^{(t+1)},\bm{x})\right]\nonumber\\ 
\quad =  \|\bm{h}\|^2 - \frac{4\mu}{m} \sum_{i=1}^m \bm{a}_i^T\bm{h}\left(\frac{|\bm{a}_i^T\bm{z}^{(t)}|^2 - |\bm{a}_i^T\bm{x}|^2 - \eta_i}{\bm{a}_i^T\bm{z}^{(t)}}\right)\mathbbm{1}_{\mathcal{E}^{i}_{t}}\nonumber\\
\quad\quad +\> \frac{4\mu^2}{m}\sum_{i=1}^m \|\bm{a}_i\|^2 \left(\frac{|\bm{a}_i^T\bm{z}^{(t)}|^2 - |\bm{a}_i^T\bm{x}|^2 - \eta_i}{\bm{a}_i^T\bm{z}^{(t)}}\right)^2 \mathbbm{1}_{\mathcal{E}^{i}_{t}}\nonumber\\
\quad =  \|\bm{h}\|^2 - \frac{4\mu}{m} \sum_{i=1}^m |\bm{a}_i^T\bm{h}|^2\left(1 + \frac{\bm{a}_i^T\bm{x}}{\bm{a}_i^T\bm{z}^{(t)}}\right)\mathbbm{1}_{\mathcal{E}^{i}_{t}} + \frac{4\mu}{m}\sum_{i=1}^m \frac{(\bm{a}_i^T\bm{h})\eta_i}{\bm{a}_i^T\bm{z}^{(t)}}\mathbbm{1}_{\mathcal{E}^{i}_{t}}\nonumber\\
\quad \quad +\> \frac{4\mu^2}{m}\sum_{i=1}^m \|\bm{a}_i\|^2|\bm{a}_i^T\bm{h}|^2\left(1 + \frac{\bm{a}_i^T\bm{x}}{\bm{a}_i^T\bm{z}^{(t)}}\right)^2 \mathbbm{1}_{\mathcal{E}^{i}_{t}} + \frac{4\mu^2}{m}\sum_{i=1}^m \frac{\eta_i^2\|\bm{a}_i\|^2}{|\bm{a}_i^T\bm{z}^{(t)}|^2}\mathbbm{1}_{\mathcal{E}^{i}_{t}} \nonumber\\
\quad \quad -\> \frac{8\mu^2}{m}\sum_{i=1}^m  \frac{(\bm{a}_i^T\bm{h})\eta_i}{\bm{a}_i^T\bm{z}^{(t)}}\left(1 + \frac{\bm{a}_i^T\bm{x}}{\bm{a}_i^T\bm{z}^{(t)}}  \right)\|\bm{a}_i\|^2\mathbbm{1}_{\mathcal{E}^{i}_{t}}\label{eq:noisy_rec}
\end{IEEEeqnarray}
To shorten the expressions, we will resort to denoting $\bm{z}^{(t)}$ by $\bm{z}$ in the remainder of this proof. The follow lemma is the noisy-case counterpart of Lemma~\ref{lem:1}.

\begin{lemma}\label{lem:4}
For any $\bm{h}$, $\bm{z}$ and $\bm{x}$, for any $\epsilon,\delta > 0$, there exist constants $C, c_0, c_1>0$ such that if $m > c_1n\delta^{-2}\log (1/\delta)$, then
$$\frac{1}{m} \sum_{i=1}^m |\bm{a}_i^T\bm{h}|^2\mathbbm{1}_{\mathcal{E}^{i}_{t}}\geq (1 - \zeta_1 - \zeta_2 - \zeta_3 - 2\delta )\|\bm{h}\|^2,$$
with probability $1 - C\exp(-c_0m\delta^2)$, for appropriately redefined $\zeta_1, \zeta_2, \zeta_3$. 
\end{lemma}
\begin{proof}
We have for any $i$,
\begin{IEEEeqnarray*}{rCl}
\mathbbm{1}_{\mathcal{E}^{i}_{3}} = \mathbbm{1}_{\left\{ y_i \leq \alpha^2_x\left(\frac{1}{m}\sum_{i=1}^my_i\right)\right\}} & \geq & \mathbbm{1}_{\left\{ |\bm{a}_i^T\bm{x}|^2 + |\eta_i| \leq \alpha^2_x \left((1-0.01)\|\bm{x}\|^2 - \frac{1}{m}\|\bm{\eta}\|_1\right) \right\}}\\
& \geq & \mathbbm{1}_{\left\{ |\bm{a}_i^T\bm{x}|^2 + \|\bm{\eta}\|_{\infty} \leq \alpha^2_x \left(0.99\|\bm{x}\|^2 - \|\bm{\eta}\|_{\infty}\right) \right\}}\\
& \geq & \mathbbm{1}_{\left\{ |\bm{a}_i^T\bm{x}|^2 + \epsilon_{\eta}\|\bm{x}\|^2 \leq \alpha^2_x \left(0.99\|\bm{x}\|^2 - \epsilon_{\eta}\|\bm{x}\|^2\right) \right\}}\\
& = & \mathbbm{1}_{\left\{ |\bm{a}_i^T\bm{x}|^2 \leq \left(\alpha^2_x \left(0.99 - \epsilon_{\eta}\right) - \epsilon_{\eta}\right)\|\bm{x}\|^2 \right\}}
\end{IEEEeqnarray*}
where the inequalities follow by \eqref{eq:ver539} and the assumption $\|\bm{\eta}\|_{\infty}\leq \epsilon_{\eta}\|\bm{x}\|^2$. 
Thus, we get that the RHS of the above expression is similar to $\mathbbm{1}_{ \left\{\frac{|\bm{a}_i^T\bm{x}|}{\|\bm{x}\| }\leq \widetilde{\alpha}_x \right\} }.$ Then, by applying similar arguments as in Lemma~\ref{lem:1}, with appropriately redefined $\zeta_1$, $\zeta_2$ and $\zeta_3$, we can conclude that the following holds with high probability,
\begin{IEEEeqnarray*}{rCl}
\frac{1}{m} \sum_{i=1}^m |\bm{a}_i^T\bm{h}|^2\mathbbm{1}_{\mathcal{E}^{i}_{t}} & \geq &  (1 - \zeta_1 - \zeta_2 - \zeta_3 - 2\delta) \|\bm{h}\|^2.
\end{IEEEeqnarray*}
\end{proof}

The following chain of inequalities, continuing from \eqref{eq:noisy_rec}, uses the fact that $\|\bm{a}_i\|^2\leq 6n$ holds with high probability, and $$\left|1 + \frac{\bm{a}_i^T\bm{x}}{\bm{a}_i^T\bm{z}^{(t)}}\right|\leq \left(1 + \frac{\sqrt{\alpha^2_x(1.01 + \epsilon_{\eta})  +\epsilon_{\eta}}(1 + \delta_{\init})}{\alpha_z^{\lb}}\right) = \widetilde{c}_4,$$ which also holds with high probability by the definition of $\mathcal{E}^i_t$, by $\|\bm{h}\|\leq \delta_{\init}\|\bm{z}\|$ and by the assumption on the noise $\|\bm{\eta}\|_{\infty}\leq \epsilon_{\eta}\|\bm{x}\|^2$. The Cauchy-Schwarz inequality is applied in terms involving $(\bm{a}_i^T\bm{h})\eta_i$. 
\begin{IEEEeqnarray}{l}
\mathbb{E}_{i_{t}}\left[\dist^2(\bm{z}^{(t+1)},\bm{x})\right] \nonumber\\
\quad \leq  \|\bm{h}\|^2 - 4\mu(1 - \zeta_1 - \zeta_2 - \zeta_3 - 2\delta )\|\bm{h}\|^2  -\frac{4\mu}{m}\sum_{i=1}^m |\bm{a}_i^T\bm{h}|^2\left(\frac{\bm{a}_i^T\bm{x}}{\bm{a}_i^T\bm{z}}\right)\mathbbm{1}_{\mathcal{E}^{i}_{t}}\nonumber \\
\quad\quad +\> \frac{4\mu}{\alpha_z^{\lb}} \frac{\|\bm{\eta}\|}{\sqrt{m}\|\bm{z}\|} \|\bm{h}\| + 4\mu^2 n c_4 (1+\delta) \|\bm{h}\|^2 + \frac{24\mu^2n}{(\alpha_z^\lb)^2} \frac{\|\bm{\eta}\|^2}{m\|\bm{z}\|^2}\nonumber\\ 
\quad\quad +\> \frac{48\mu^2n\widetilde{c}_4\sqrt{1+\delta}}{\alpha_z^\lb} \|\bm{h}\| \frac{\|\bm{\eta}\|}{\sqrt{m}\|\bm{z}\|}
\end{IEEEeqnarray}

The third term in the RHS can be shown to be no more than a small constant times $\|\bm{h}\|^2$ in a similar manner as in Lemma~\ref{lem:2}. Thus, we get by collecting all constants for brevity,
\begin{IEEEeqnarray}{rCl}
\mathbb{E}_{i_{t}}\left[\dist^2(\bm{z}^{(t+1)},\bm{x})\right] & \leq & \|\bm{h}\|^2 - 4\mu C\|\bm{h}\|^2 + \mu C' \frac{\|\bm{\eta}\|}{\sqrt{m}\|\bm{z}\|} \|\bm{h}\| + \mu^2 n C'' \|\bm{h}\|^2 \nonumber\\
&&\quad +\> \mu^2 n C''' \frac{\|\bm{\eta}\|^2}{m\|\bm{z}\|^2} + \mu^2 n C'''' \|\bm{h}\| \frac{\|\bm{\eta}\|}{\sqrt{m}\|\bm{z}\|}. \label{eq:noisy_simplified}
\end{IEEEeqnarray}

\subsubsection{Regime I}
Regime I is defined to be the condition $\frac{c_3\|\bm{\eta}\|}{\sqrt{m}\|\bm{z}\|} \leq \|\bm{h}\| \leq c_4\|\bm{x}\|$ where $c_3$ is some large constant. In this regime, it is possible to show that the distance to the solution decreases geometrically, as in the noiseless case.

Because $\frac{\|\bm{\eta}\|}{\sqrt{m}\|\bm{z}\|} \leq \frac{1}{c_3}\|\bm{h}\|$ in regime 1, we can substitute this in \eqref{eq:noisy_simplified} and choose an appropriate $\mu =\Theta \left(\frac{1}{n}\right)$, assuming $c_3$ is large enough, to arrive at
\begin{IEEEeqnarray}{rCl}
\mathbb{E}_{i_{t}}\left[\dist^2(\bm{z}^{(t+1)},\bm{x})\right] & \leq & \left(1 - \frac{\rho}{n}\right)\|\bm{h}\|^2. \nonumber
\end{IEEEeqnarray}

Thus, we have linear convergence in Regime 1.

\subsubsection{Regime II}
If $\bm{h}$ is not in Regime I, we have $\|\bm{h}\| \leq \frac{c_3\|\bm{\eta}\|}{\sqrt{m}\|\bm{z}\|}$. While the distance to the solution is not guaranteed to decrease in this regime, it can be shown that by taking one step of ITWF, the expected next iterate either stays in Regime II or goes to Regime I, but not beyond, thus resulting in a bounded error, as given by Theorem~\ref{thm:2}. We employ similar arguments as in \cite{CheCan15}.  

The expected norm of the step is
\begin{IEEEeqnarray*}{rCl}
\mathbb{E}_{i_t}\left[\left\|\mu \nabla\ell(y_{i_t},|\bm{a}_{i_t}^T\bm{z}^{(t)}|^2)\cdot \mathbbm{1}_{\mathcal{E}^{i_t}_{t}} \right\|\right] & = & \mathbb{E}_{i_t}\left[\left\|\mu \left(\frac{|\bm{a}_i^T\bm{z}|^2 - |\bm{a}_i^T\bm{x}|^2 - \eta_i}{\bm{a}_i^T\bm{z}}\right)\bm{a}_i\cdot \mathbbm{1}_{\mathcal{E}^{i_t}_{t}} \right\|\right]\\
& \lesssim & \mu \sqrt{n} \left(\|\bm{h}\| + \frac{\|\bm{\eta}\|}{\sqrt{m}\|\bm{z}\|}\right) \\
& \lesssim & \frac{1}{\sqrt{n}}\left(\|\bm{h}\| + \frac{\|\bm{\eta}\|}{\sqrt{m}\|\bm{z}\|}\right).
\end{IEEEeqnarray*} 

If $\|\bm{h}\| \leq \frac{c_3\|\bm{\eta}\|}{\sqrt{m}\|\bm{z}\|}$, we have that 
\begin{IEEEeqnarray*}{rCl}
\mathbb{E}_{i_t}[\dist^2(\bm{z}^{(t+1)},\bm{x})] & \lesssim & \|\bm{h}\|^2 + \frac{1}{n}\left(\|\bm{h}\|^2 + \frac{\|\bm{\eta}\|^2}{m\|\bm{z}\|^2}\right) \\
& \leq & C\frac{\|\bm{\eta}\|^2}{m\|\bm{z}\|^2}.
\end{IEEEeqnarray*}

Thus, if $C$ is small, we stay in Regime II, else we have that
$$C\frac{\|\bm{\eta}\|^2}{m\|\bm{z}\|^2} \leq C\frac{\|\bm{\eta}\|^2_{\infty}}{\|\bm{z}\|^2} \leq \widetilde{C}\|\bm{x}\|^2,$$
thus ending up in Regime I, assuming $\frac{\|\bm{\eta}\|_{\infty}}{\|\bm{x}\|^2}=\epsilon_{\eta}$ is sufficiently small.

This concludes the proof of Theorem~\ref{thm:2}.\qed

\end{document}